\pdfoutput=1
\documentclass{toc}

\usepackage{enumitem}
\setlist{nolistsep}

\usepackage{complexity}
\usepackage{float}
 \floatstyle{ruled}
 \newfloat{proof-system}{tph}{lop}
\floatname{proof-system}{Interactive Proof System}

\usepackage{caption}
\usepackage[nocompress]{cite}

\usepackage{mathtools}

\newcommand{\hgate}{{\sf H}}
\newcommand{\pgate}{{\sf P}}
\newcommand{\rgate}{{\sf T}}
\newcommand{\tgate}{{\sf T}}
\newcommand{\xgate}{{\sf X}}
\newcommand{\ygate}{{\sf Y}}
\newcommand{\zgate}{{\sf Z}}
\newcommand{\cnot}{{\sf CNOT}}
\providecommand{\abs}[1]{\lvert#1\rvert}

\newcommand{\QPIP}{\mathsf{QPIP}}

\newcommand{\cC}{\mathcal{C}}

\newcommand{\calS}{\mathcal{S}}

\newcommand{\norm}[1]{\left\lVert#1\right\rVert}					%
\newcommand{\ketbra}[2]{\ket{#1}\bra{#2}}			%
\newcommand{\egoketbra}[1]{\ketbra{#1}{#1}}			%

\newcommand{\stepref}[2]{\hyperref[#2]{Step~\ref{#1}\ref{#2}}}

\input{qcircuit}
\newcommand{\dw}[1][-1]{\ar @{--} [0,#1]}
\newcommand{\cgate}[1]{*+<.6em>{#1} \POS ="i","i"+UR;"i"+UL **\dir{-};"i"+DL **\dir{-};"i"+DR **\dir{-};"i"+UR **\dir{-},"i" \cw}
\tocdetails{%
volume=14, number=11, year=2018, firstpage=1,
received={October 3, 2016},
revised={October 27, 2017},
published={June 11, 2018},
doi={10.4086/toc.2018.v014a011},
title={How to Verify a Quantum Computation},
author={Anne Broadbent},
plaintextauthor={Anne Broadbent},
acmclassification={F.1.3},
amsclassification={68Q15, 81P68},
keywords={complexity theory, cryptography, interactive proofs, quantum computing, quantum interactive proofs, quantum cryptography}
}

\begin{document}

\begin{frontmatter}[classification=text]
\title{How to Verify a Quantum Computation}

\author[broadbent]{Anne Broadbent\thanks{This material is based upon work supported by the Air Force Office of Scientific Research under award number FA9550-17-1-0083, Canada's NSERC, and the University of Ottawa's Research Chairs program.
}}

\begin{dedication}
  To my daughter \'Emily, on the occasion of her first birthday.
\end{dedication}

\begin{abstract}
We give a new theoretical solution to a leading-edge experimental challenge, namely to the verification of quantum computations in the regime of high computational complexity.
Our results are given in the language of quantum interactive proof systems. Specifically, we show that any language in $\BQP$ has a quantum interactive proof system with a polynomial-time classical verifier (who can also prepare random single-qubit pure states), and a quantum polynomial-time prover. Here, soundness is unconditional---\ie,~it holds even for computationally unbounded provers. Compared to prior work achieving similar results, our technique does not require the encoding of the input or of the computation; instead, we rely on encryption of the input (together with a method to perform computations on encrypted inputs), and show that the random choice between three types of input (defining a \emph{computational run}, versus two types of \emph{test runs}) suffices. Because the overhead is very low for each run (it is linear in the size of the circuit), this shows that verification could be achieved at minimal cost compared to performing the computation.  As a proof technique, we use a reduction to an entanglement-based protocol; to the best of our knowledge, this is the first time this technique has been used in the context of verification of quantum computations, and it enables a relatively straightforward analysis.
\end{abstract}

\tocacm{F.1.3}
\tocams{68Q15, 81P68}
\tockeywords{complexity theory, cryptography, interactive proofs, quantum computing, quantum interactive proofs, quantum cryptography}

\iffalse %
%
%
%
%
%
%
%
%
%
%
%
%
%
%
%
%
%
%
%
%
%

\tocarxivcategory{quant-ph}

\fi %

\end{frontmatter}

\section{Introduction}

Feynman~\cite{Fey82} was the first to point out that quantum computers, if built, would be able to perform quantum simulations (\ie,~to compute the predictions of quantum mechanics; which is widely believed to be classically intractable). But this immediately begs the question: if the output of a quantum computation cannot be predicted, how do we know that it is correct? Conventional wisdom would tell us that we can rely on testing \emph{parts} (or scaled-down versions) of a quantum computer---conclusive results would then extrapolate to the larger system. But this is somewhat unsatisfactory, since we may not rule out the hypothesis that, at a large scale, quantum computers behave unexpectedly. A different approach to the verification of a quantum computation would be to construct a number of quantum computers based on different technologies (\eg,~with ionic, photonic, superconducting and/or solid state systems), and to accept the computed predictions if the experimental results agree. Again, this is still somewhat unsatisfactory, as a positive outcome does not confirm the correctness of the output, but instead confirms that the various large-scale devices behave similarly on the given instances.

This problem, though theoretical in nature~\cite{AV14}, is already appearing as a major experimental challenge. One of the outstanding applications for the verification of quantum systems is in quantum chemistry, where the current state-of-the-art  is that the inability to verify quantum simulations is  much more the norm than the exception~\cite{GH05}. Any theoretical advance in this area could have dramatic consequences on applications of quantum chemistry simulations, including the potential to revolutionize drug discovery.
Another case where experimental techniques are reaching the limits of classical verifiability is in
the Boson Sampling problem~\cite{AA11}, where the process of verification has been raised as a fundamental objection to the viability of experiments~\cite{GKAE13} (fortunately, these claims are refuted~\cite{AA14}, and progress was made in the experimental verification~\cite{SVB+14}).

As mere classical probabilistic polynomial-time\footnote{\Ie,~assuming humans can flip coins and execute classical computations that take time polynomial in $n$ to solve on inputs of size~$n$.} individuals, we appear to be in an impasse: how can we validate the output of a quantum computation?\footnote{Assuming the widely-held belief that $\BQP \neq \BPP $, \ie,~that quantum computers are indeed more powerful than classical computers.}
For some problems of interest in quantum computing (such as factoring and search), a claimed solution can be efficiently verified by a classical computer. However, current techniques do not give us such an efficient verification procedure for the \emph{hardest} problems that can be solved by quantum computers (such problems are known as $\BQP$-complete, and include the problem of approximating the Jones polynomial~\cite{AJL06}).
Here, we propose a solution based on \emph{interaction}, viewing an experiment not in the traditional, static, predict-and-verify framework, but as an interaction between an experimentalist and a quantum device.  In the context of theoretical computer science, it has been established for quite some time that interaction between a  probabilistic polynomial-time \emph{verifier} and a computationally unbounded \emph{prover} allows the verification of a class of problems \emph{much} wider than what static proofs allow.\footnote{This is the famous $\IP=\PSPACE$ result~\cite{LFKN90,Sha92}.}

Interactive proof systems traditionally model the prover as being all-powerful (\ie,~computationally unbounded).\footnote{Notable exceptions include~\cite{CL95, GKR08}.} For our purposes, we restrict the prover to being a ``realistic'' quantum device, \ie,~we model the prover as a quantum polynomial-time machine. Our approach equates the verifier with a classical polynomial-time machine, augmented with \emph{extremely} rudimentary quantum operations, namely of being able to prepare random single pure-state qubits (chosen among a specific set, see \expref{Section}{sec:QPIP-definitions}). Our verifier does not require any quantum memory or quantum processing power.
Without loss of generality, the random quantum bits can be sent in the first round, the balance of the interaction and verifier's computation being classical.
Formally, we present our results in terms of an \emph{interactive proof system}, showing that in our model, it is possible to devise a
\emph{quantum-prover interactive proof system}
for all problems solvable (with bounded error) in quantum polynomial time.

\subsection{Related work}
\label{sec:related-work}

The complexity class $\QPIP$, corresponding to quantum-prover interactive proof systems, was originally defined by Aharonov, Ben-Or and Eban~\cite{ABE10}, who, using techniques from~\cite{BCG+06}, showed that $\BQP = \QPIP$ for a verifier with the capacity to perform  quantum computations on a constant-sized quantum register (together with polynomial-time classical computation). The main idea of~\cite{ABE10} is to encode the input into a \emph{quantum authentication code}~\cite{BCG+02}, and to use interactive techniques for \emph{quantum computing on authenticated data} in order to enable verification of a  quantum computation.
  This result was revisited in light of foundations of physics in~\cite{AV14}, and the protocol was also shown secure in a scenario of \emph{composition}~\cite{BGS13}.

In a different line of research, Kashefi and Fitzsimons~\cite{FK12} consider a measurement-based approach to the problem, giving a scheme that  requires the verifier to prepare only random single qubits: the main idea is to encode the computation into a larger one which includes a verification mechanism, and to execute the resulting computation using blind quantum computing~\cite{BFK09}. Thus, success of the encoded computation can be used to deduce the correctness of actual computation. A small-scale version of this protocol was implemented in quantum optics~\cite{BFKW13}. Further work by Kapourniotis,  Dunjko and Kashefi~\cite{KDK15} shows how to combine the~\cite{ABE10} and~\cite{FK12} protocols in order to reduce the quantum communication overhead; Kashefi and  Wallden~\cite{KW15} also show how to reduce the overhead of~\cite{FK12}.

To the best of our knowledge,  the proof techniques in these prior works appear as sketches only, or are cumbersome. In particular, the approach that uses quantum authentication codes~\cite{ABE10} is based on~\cite{BCG+06}. However, the full proof of security for~\cite{BCG+06} never appeared. Although~\cite{ABE10} makes significant progress towards closing this gap, it provides only a sketch of how the soundness is established in the interactive case (see, however, the very recent~\cite{ABEM17}). A full proof of soundness for~\cite{ABE10} follows from~\cite{BGS13}, however the proof is very elaborate and phrased in terms of a rather different cryptographic task (called  ``quantum one-time programs'').
In terms of the measurement-based approach, note that a proposed protocol for verification in~\cite{BFK09} was deemed incomplete~\cite{FK12}, but any gaps were addressed in~\cite{FK12}. In this case, however, the protocol (and proof) are very elaborate, and to the best of our knowledge, remain unpublished.\footnote{This work has been published as~\cite{FK17} during the review period of the current paper.} Note, however that follow-up work has appeared in peer-reviewed form~\cite{KDK15,KW15}, and that these works consider the more general problem of verification for \emph{quantum} inputs and outputs.

A related line of research also studies the problem of verification with a client that can perform only single-qubit measurements~\cite{HM15}; the case of untrusted devices is also considered in~\cite{HH16}. In sharp contrast to these approaches, Reichardt, Unger  and Vazirani~\cite{RUV13} show that it is possible to make the verifier \emph{completely} classical, as long as we postulate \emph{two} non-communicating entangled
provers.  (This
could be enforced, for instance, by space-like separation such that
communication between the provers would be forbidden by the limit on the
speed of light.)   %
The main technique used  is a \emph{rigidity theorem} which, provided that the provers pass a certain number of tests, gives the verifier a tight classical control on the quantum provers. Very recently, Coladangelo, Grilo,  Jeffery, and  Vidick~\cite{CGJV17} have used the techniques described here to achieve efficient schemes for verifying quantum computations in the model of a classical verifier and two entangled provers.

\subsection{Contributions}

Our main contributions are a new, simple  quantum-prover interactive proof system for~$\BQP$, with a verifier whose quantum power is limited to the random preparation of single-qubit pure states, together with a new proof
technique.  %

\paragraph{New protocol.}
All prior approaches to the verification of quantum computations required some type of encoding (either of the input or of the computation), or otherwise had the verifier perform part of the computation. In contrast,  our protocol achieves soundness via the verifier's random choice of different types of runs. This is a typical construction in interactive proofs, and in some sense it is surprising that it is used here for the
first time in the context of verifying quantum computations. According to the new protocol, the overhead required for verification can be reduced to repetition of a very simple protocol (with overhead at most linear compared to performing the original computation), and thus may lead to implementations sooner than expected  (in general, it is much easier to repeat an experiment using different initial settings, than to run a single, more complex experiment!).

\paragraph{New proof technique.} In order to prove soundness, we use the proof technique of a reduction to an ``entanglement-based'' protocol. This proof technique  originates from Shor and Preskill~\cite{SP00} and has been used in a number of quantum cryptographic scenarios, \eg,~\cite{DFSS05, DFPR14,FBS+14}. To the best of our knowledge, this is the first time that this technique is used in the context of the verification of quantum computations; we show how the technique provides a much-needed succinct and convincing method to prove soundness. In particular, it allows us to reduce the analysis of an interactive protocol to the analysis of a non-interactive one, and to formally delay the verifier's choice of run until \emph{after} the interaction with the prover.

Furthermore, this work unifies the two distinct approaches given above, (one based on quantum authentication codes and the other on measurement-based quantum computing). Indeed, one can view our protocol as performing a very basic type of quantum computing on authenticated data~\cite{BGS13}; with hidden gates being executed via a computation-by-teleportation process~\cite{GC99} that is reminiscent of measurement-based quantum computation, and thus of blind quantum computation~\cite{BFK09}.\looseness=-1

On the conceptual front, this work  focuses on the \emph{simplest possible} way to achieve a quantum-prover interactive proof system.
Via this process, we have
further emphasized  links between various
concepts. %
\begin{enumerate}
\item \textbf{A link between input privacy and verification.} Prior results~\cite{ABE10, BFK09, BGS13, FK12} all happened to provide both \emph{privacy} of a quantum computation and its \emph{verification} (one notable exception being the recent~\cite{FH15}). Here, we make this link explicit, essentially starting from input privacy and constructing a verifiable scheme (this was also done, to a certain extent in~\cite{BFK09, FK12}).

    \item \textbf{A link between fault-tolerant quantum computation and cryptography.} Prior results~\cite{BCG+06, ABE10, BGS13} used constructions inspired by fault-tolerant quantum computation. Here, we make the link even more explicit by using single-qubit gate gadgets that are adaptations of the gate gadgets used in fault-tolerant quantum computation.  Furthermore, our results also emphasize how the ubiquitous technique of  ``tracking the Pauli frame'' from fault-tolerant quantum computation can be re-phrased in terms of keeping track of an encryption key.\looseness=-1

\item \textbf{A link between entanglement and parallelization.} It is known that entanglement can reduce the number of rounds in quantum interactive proof systems~\cite{KKMV08}; a consequence of our entanglement-based protocol is that we can parallelize our interactive proof system to a single round, as long as we are willing to allow the prover to share entanglement with the verifier, and to perform adaptive measurements.
\end{enumerate}

\subsection{Overview of techniques}
\label{sec:overview-techniques}

The main idea for our quantum-prover interactive proof system is that the verifier chooses randomly to interact with the prover in one of three runs.
Among these runs, one is the \emph{computation} run, while the two others are \emph{test} runs. In an honest interaction, the output of the computation run is the result (a single bit) of evaluating the given quantum circuit. The test runs are used to detect a deviating prover; there are two types of test runs: an \emph{$\xgate$-test} and a \emph{$\zgate$-test}.
Intuitively (and formally proved in \expref{Section}{sec:EPR-based-QPIP}), we  see that the prover cannot distinguish between all three runs.  Thus, his  strategy must be invariant over the different runs. %
It should be clear now how this work links \emph{input privacy} with verification:  by varying the input to the computation, the verifier differentiates between test and computation runs; by input privacy, however, the prover cannot identify the type of run and thus any deviation from the prescribed protocol has a chance of being detected.

In more details, the runs have the following properties (from the point of view of the verifier)
\begin{itemize}
\item \textbf{Computation run.} In a computation run, the prover executes the target circuit on input~$\ket{0}^{\otimes n}$.
\item \textbf{$\xgate$-test run.} In an $\xgate$-test run, the prover executes the identity circuit on input $\ket{0}^{\otimes n}$. At the end of the computation, the verifier verifies that the result is~$0$.  This test also contains internal checks for cheating within the protocol.
\item \textbf{$\zgate$-test run.}  In a $\zgate$-test run, the prover executes the identity circuit on input $\ket{+}^{\otimes n}$.  This test run is used only as an internal check for cheating within the protocol.
\end{itemize}

In order for the prover to execute the above computations without being able to distinguish between the runs, we use a technique inspired by \emph{quantum computing on encrypted data (QCED)}~\cite{FBS+14,Bro15}: the input qubits are encrypted with a random Pauli, as are auxiliary qubits that are used to drive the computation.

Viewing the target computation as a sequence of gates in the universal
set of gates
$\{\xgate, \zgate, \hgate, \cnot, \rgate\}$ (see \expref{Section}{sec:prelim-notation} for notation), the task we face is, in the computation run, to perform these logical gates on encrypted quantum data. Furthermore, the $\xgate$- and $\zgate$-test runs should (up to an encryption key), leave the quantum wires in the $\ket{0}^{\otimes n}$ or $\ket{+}^{\otimes n}$ state, respectively.
Performing Pauli gates in this fashion is straightforward, as this can be done by adjusting the encryption key (in the computation run only). As we show in \expref{Section}{sec:CNOT-gadget}, the  $\cnot$ gate can be executed directly (since it does not have any effect on the wires for the test runs). The  $\rgate$-gate (\expref{Section}{sec:R-gadget}) is performed using a construction (``gate gadget'') inspired both by QCED and fault-tolerant quantum computation~\cite{BMP+00} (see also~\cite{BJ15}); the $\rgate$-gate gadget involves the use of an auxiliary qubit and  classical interaction. The $\hgate$ is performed thanks to an identity involving the $\hgate$ and $\pgate$ (\expref{Section}{sec:H-gadget}). Note that $\pgate$ can be accomplished as $\rgate^2$.

In order to prove soundness, we consider any general deviation of the prover, and show that such deviation can be mapped to an attack on the measured wires only, corresponding to an honest run of the protocol (without loss of generality, we can also delay all measurements until the end of the protocol).
Furthermore, because the computation is performed on encrypted data, by the \emph{Pauli twirl}~\cite{DCEL09}, this attack can be described as a convex combination of Pauli attacks on the measured qubits. Since all measurements are performed in the computational basis, $\zgate$~attacks are obliterated, and thus the only family of attacks of concern consists in  $\xgate$- and $\ygate$-gates applied to various measured qubits; these act as bit flips on the corresponding classical output. We show that the combined effect of test runs is to detect \emph{all} such attacks; this allows us to bound the probability that the verifier accepts a \emph{no}-instance. Since only $\xgate$ and $\ygate$ attacks require detection, one may wonder why we use also a $\zgate$-test run. The answer to this question lies in the implementation of the $\hgate$-gate: while its net effect is to apply the identity in the test runs,
its  internal workings  %
temporarily \emph{swap}
the roles of the $\xgate$- and $\zgate$-test
runs; %
thus the $\zgate$-test runs are also %
used to detect $\xgate$ and $\ygate$ errors.

Finally, some words on showing indistinguishability between the test and computation runs. This is done by showing that the verifier can delay her choice of  run (computation, $\xgate$- or $\zgate$-test) until \emph{after} the interaction with the prover is complete. This is accomplished via an entanglement-based protocol, where the verifier's messages to the prover consist in only half-EPR pairs, as well as classical random bits. These messages are identical in both the test and computation runs; as the verifier decides on the type of run only \emph{after} having the interacted with the prover. Depending on this choice, the verifier performs measurements on the system returned by the prover, resulting in the desired effect.

\subsection{Open problems}
The main outstanding open problem is the verifiability of a quantum computation with a  \emph{classical} verifier, interacting with a \emph{single} quantum polynomial-time prover. In this context, we
make   %
a few
observations. %
\begin{itemize}
\item  If the prover is unbounded, there exists a quantum interactive proof system for $\BQP$, since
$\QIP (=\PSPACE) =\IP$.\footnote{$\QIP =\PSPACE$ is due to~\cite{JJUW10}.}
\item  If $\P= \BQP$, there is a trivial quantum interactive proof system.
\item  One possible approach would be to relax the definition to require only \emph{computational} soundness (following the lines of Brassard, Chaum and Cr\'epeau~\cite{BCC88}, this would lead to a quantum interactive \emph{argument}). This approach seems promising, especially if we consider a computational assumption that is \emph{post-quantum} secure. If, via its interaction with the prover, a classical verifier accepts, then we can conclude that either the verifier performed the correct computation \emph{or} the prover has broken the computational assumption.
\end{itemize}

\subsection{Organization}
The remainder of this paper is organized as follows. \expref{Section}{sec:prelim} presents some preliminary notation and background.  \expref{Section}{sec:QPIP-definitions} defines quantum-prover interactive proofs and states our main theorem. \expref{Section}{sec:Interactive-Proof-aux} describes the interactive proof system, for which we show completeness (\expref{Section}{sec:completeness}), and soundness (\expref{Section}{sec:soundness}).

\section{Preliminaries}
\label{sec:prelim}
\subsection{Notation}
\label{sec:prelim-notation}

We assume the reader is familiar with the
basics of quantum information~\cite{NC00}.
We use the following well-known qubit gates
\begin{align}
  &\xgate: \ket{j} \mapsto \ket{j \oplus 1}\,,\\
  &\zgate: \ket{j} \mapsto (-1)^j\ket{j}\,,\\
  \text{Hadamard}\quad & \hgate: \ket{j} \mapsto \frac{1}{\sqrt{2}}(\ket{0}+ (-1)^j\ket{1})\,,\\
  \text{phase gate} \quad&  \pgate: \ket{j} \mapsto i^{j}\ket{j}\,,\\
  \text{$\pi/8$ rotation}\quad & \rgate: \ket{j} \mapsto e^{(i\pi/4)^j}\ket{j})\,,\qquad\text{and the}\\[.5ex]
  \text{two-qubit controlled-not}\quad & \cnot: \ket{j}\ket{k} \mapsto \ket{j} \ket{j\oplus k}\,.
\end{align}
Let $\ygate = i\xgate\zgate$.  We denote by $\mathbb{P}_n$ the set of
$n$-qubit Pauli operators, where $P \in \mathbb{P}_n$ is given by
$P = P_1 \otimes P_2 \otimes \cdots \otimes P_n$ where
$P_i \in \{I, \xgate, \ygate, \zgate\}$; we also denote an \emph{EPR
  pair}
\begin{equation}
  \ket{\Phi^+} = \frac{1}{\sqrt{2}}(\ket{00} + \ket{11})\,.
\end{equation}

\subsection{Quantum encryption and the Pauli twirl}
\label{sec:prelim-QOTP}
The quantum one-time pad encryption maps a single-qubit system $\rho$ to
\begin{equation}
  \frac{1}{4}\sum_{a,b \in \{0,1\}} \xgate^a \zgate^b \rho \zgate^b \xgate^a = \frac{I}{2}\,;
\end{equation}
its generalization to $n$-qubit systems is straightforward~\cite{AMTW00}.
Here, we take $(a,b)$ to be the classical private \emph{encryption key}. Clearly, this scheme provides information-theoretic security, while allowing decryption, given knowledge of the~key.
A useful observation is that if we have an \emph{a priori} knowledge of the quantum operator~$\rho$, then it may not be necessary to encrypt it with a full quantum one-time pad (\eg,~if the state corresponds to a pure state of the form $({1}/{\sqrt{2}})(\ket{0} + e^{i\theta}\ket{1})$, it can be encrypted with a random~$\zgate$), although there is no loss of generality in encrypting it with the full random Pauli. We use the two interpretations interchangeably.

Consider for a moment the classical one-time pad (that encrypts a plaintext message by XORing it with a random bit-string of the same length). It is intuitively clear that if an adversary (who does not know the encryption key) has access to the ciphertext only, and is allowed to modify it, then the effect of any adversarial attack (after decryption) is to probabilistically introduce bit flips in target locations. The quantum analogue of this is given by the \emph{Pauli twirl}~\cite{DCEL09}.  %

\begin{lemma}[Pauli Twirl]\label{lem:Pauli-twirl}
Let $P, P'$  $\in \mathbb{P}_n$. Then %
\begin{equation}
\frac{1}{\abs{\mathbb{P}_n}} \sum_{Q \in \mathbb{P}_n} Q^* P Q \rho Q^* P'^* Q = \begin{cases} 0,& P \neq P', \\ P\rho P^* ,& \text{otherwise}\,.\end{cases}
\end{equation}
\end{lemma}

We also obtain the classical case for the single-qubit  Pauli twirl,
alluded to above, as the following. %

\begin{lemma}[Classical Pauli Twirl]\label{lem:classical-Pauli-twirl-X} Let $c,i \in \{0,1\}$ and $P, P'$  $\in \mathbb{P}_1$.
Then %
\begin{equation}\frac{1}{2} \sum_{Q \in \{I, \xgate\}} \bra{i}  Q^* P Q \ket{c}\bra{c} Q^* P' Q \ket{i} =
 \begin{cases}
 0, & P \neq P',\\
\bra{i}   P  \ket{c}\bra{c}  P \ket{i},& \text{otherwise}.
  \end{cases}
    \end{equation}
\end{lemma}

\begin{proof}
The proof is a simple application of \expref{Lemma}{lem:Pauli-twirl},
together with the observation that $\ket{0}, \ket{1}$ are eigenstates
of $\zgate$:
\begin{align}
\frac{1}{2}\sum_{Q \in \{I, \xgate\}} \bra{i}  Q^* P Q \ket{c}\bra{c} Q^* P' Q \ket{i} &=
\frac{1}{4} \sum_{Q \in \{I, \xgate, \ygate, \zgate \}}   \bra{i}  Q^* P Q \ket{c}\bra{c} Q^* P' Q \ket{i}\\
&=\begin{cases}
 0,& P \neq P',\\
\bra{i} P\ket{c}\bra{c} P \ket{i},& \text{\emph{otherwise}}.
\end{cases}
\end{align}
\end{proof}
Working in the basis $\{\ket{+}, \ket{-}\}$, we also
obtain the following. 
\begin{lemma}\label{lem:classical-Pauli-twirl-Z} Let $c,i \in \{0,1\}$ and $P, P'$  $\in \mathbb{P}_1$. Then %
\begin{equation}\frac{1}{2} \sum_{Q \in \{I, \xgate\}} \bra{i}  Q^*H P H Q \ket{c}\bra{c} Q^*H P'H Q \ket{i} =
 \begin{cases}
 0, & P \neq P',\\
\bra{i} H  P H \ket{c}\bra{c}H  P H\ket{i},& \text{otherwise}.
  \end{cases}
    \end{equation}
\end{lemma}

\section{Definitions and statement of results}
\label{sec:QPIP-definitions}

Interactive proof systems were introduced by Babai~\cite{Bab85} and Goldwasser, Micali,
and Rackoff~\cite{GMR89}. An interactive proof system consists of an interaction between a
computationally unbounded prover and a computationally bounded probabilistic verifier.
For a language~$L$ and a string~$x$, the prover attempts to convince the verifier that $x \in L$,
while the verifier tries to determine the validity of this ``proof.''
Thus, a language~$L$ is said to have an interactive proof system if
there exists a polynomial-time verifier~$V$
with the following properties.  %
\begin{itemize}
 \item (Completeness) if~$x \in L$, there exists a prover (called an honest prover) such that the verifier accepts with probability $p \geq {2}/{3}$;
 \item (Soundness) if $x \not\in L$, no prover can convince $V$ to accept with probability $p \geq 1/3$.
 \end{itemize}
The class of languages having interactive proof systems is denoted~$\IP$.

Watrous~\cite{Wat03} defined~$\QIP$ as the quantum analogue of~$\IP$, \ie,~as the class of languages having a \emph{quantum} interactive proof system, which consists in a quantum interaction between a
computationally unbounded quantum prover and a computationally bounded quantum verifier, with the analogous completeness and soundness conditions as given above.

For our results, we are interested in the scenario of a polynomial-time prover (in the honest case), as well as  an \emph{almost-classical} verifier; that is, a verifier with the power to generate random qubits as specified by a parameter~$\calS$ (\expref{Definition}{defn:QPIP-system}). Furthermore, as a technicality, instead of considering languages, we consider
promise problems.  A promise  %
problem $\Pi = (\Pi_Y, \Pi_N)$ %
is a pair of disjoint sets of strings, corresponding to $\text{YES}$ and $\text{NO}$ instances, respectively. For a formal treatment of the model (which we specialize here to our scenario), see~\cite{Wat03}.

\begin{definition}
\label{defn:QPIP-system}
Let $\mathcal{S} = \{\calS_1, \ldots ,\calS_\ell\}$ where $\calS_i = \{\rho_1, \ldots ,\rho_{\ell_i}\}$ $(i=1, \ldots ,\ell)$ is a set of density operators.
A  $\calS$-\emph{quantum-prover Interactive Proof System   for a promise problem $\Pi = (\Pi_Y, \Pi_N)$} is an interactive proof system with a verifier $V$ that  runs in classical probabilistic polynomial time, augmented with the capacity to randomly generate states in each of  $\calS_1, \ldots ,\calS_\ell$ (upon generation, these states are immediately sent to the prover, with the index $ i \in \{1, \ldots ,\ell\}$ known to the verifier and prover, and the index $ j \in \{1, \ldots ,\ell_i\}$ known to the verifier only).
The interaction of the verifier $V$ and the prover $P$
satisfies the following conditions.
\begin{itemize}
 \item (Completeness) if~$x \in \Pi_Y$, there exists a quantum polynomial-time prover (called an honest prover) such that the verifier accepts with probability $p \geq {2}/{3}$;
 \item (Soundness) if $x \in \Pi_N$, no prover (even unbounded) can convince $V$ to accept with probability $p \geq 1/3$.
 \end{itemize}
\end{definition}
The class of promise problems having an  $\calS$-quantum interactive proof systems is denoted~$\QPIP_\calS$. Note that by standard amplification, the class $\QPIP_\calS$  is unchanged if we replace the completeness parameter $c$ and soundness parameter~$s$ by any values, as long as $c-s > {1}/{\poly(n)}$.

Comparing our definition of $\QPIP_\calS$ to the class of quantum-prover interactive proof systems ($\QPIP$) as given in~\cite{ABE10}, we note that we have made some modifications and clarifications, namely that the verifier in $\QPIP_\calS$ does not have any quantum memory and does not perform any gates ($\QPIP$ allows a verifier that stores and operates on a quantum register of a constant number of qubits), and that soundness holds against unbounded provers.

Finally, we use the %
canonical $\BQP$-complete problem~\cite{ABE10},  %
defined as follows.    %
\begin{definition} \label{defn-Q-circuit}
The input to the promise problem \textsf{Q-CIRCUIT} consists 
of a quantum circuit made of a sequence of gates, $U= U_T, \ldots, U_1$ acting on $n$ input
qubits.  (We take
these circuits to be given in the universal gateset $\{\xgate, \zgate, \hgate,
\cnot, \tgate\}.)$   %
Let
  \begin{equation}
    p(U)= \norm{\ket{0}\bra{0}\otimes \mathbb{I}_{n-1} U \ket{0^n}}^2
  \end{equation}
  be the probability of observing ``0'' as a result of a computational
  basis measurement of the $n^{\text{th}}$ output qubit, obtained by
  evaluating~$U$ on input $\ket{0^n}$.

Then define Q-CIRCUIT$=\{\mathsf{Q-CIRCUIT}_\mathsf{YES}, \mathsf{Q-CIRCUIT}_\mathsf{NO}\}$ with %
\begin{align}
\mathsf{Q-CIRCUIT}_\mathsf{YES} &:  p(U)  \geq 2/3\,, \\
\mathsf{Q-CIRCUIT}_\mathsf{NO} &:  p(U)  \leq 1/3\,.
\end{align}
\end{definition}
We can now formally state our main theorem.
\begin{theorem}[Main Theorem]
\label{thm:main-QPIP}
Let
\begin{equation}
  \mathcal{S} = \bigl\{ \{\ket{0}, \ket{1}\}, \{\ket{+}, \ket{-}\}, \{\pgate\ket{+}, \pgate\ket{-}\},
  \{\tgate\ket{+}, \tgate\ket{-}, \pgate\tgate\ket{+}, \pgate\tgate\ket{-}\}\bigr\}\,.
\end{equation}
Then $\BQP = \QPIP_\calS$.
\end{theorem}

\section{Quantum-prover interactive proof system}
\label{sec:Interactive-Proof-aux}

In order to prove \expref{Theorem}{thm:main-QPIP}, we give   an interactive proof system (see \textbf{\expref{Interactive Proof System}{prot:interactive-proof}}).
This protocol uses the various gate gadgets as described in \expref{Sections}{sec:Pauli-gadget}--\ref{sec:H-gadget}. Completeness is  studied in \expref{Section}{sec:completeness} and soundness is proved in \expref{Section}{sec:soundness}.

\begin{proof-system} \caption{\label{prot:interactive-proof} Verifiable quantum computation with trusted auxiliary states}
Let $\mathcal{C}$ be given as an $n$-qubit quantum circuit in the universal gateset $\xgate, \zgate, \cnot, \hgate, \rgate$.

\begin{enumerate}
\item %
The verifier randomly chooses to execute one of the following three runs (but does not inform the prover of this choice).

\begin{enumerate}[label=\Alph*.]
\item   \textbf{Computation Run}\label{step:CR}
\begin{enumerate}[label*=\arabic*.,ref=\arabic*]
\item The verifier encrypts input $\ket{0}^{\otimes n}$ %
and sends~the input to $P$.
\item \label{step:comp-aux}The verifier sends auxiliary qubits required for the $\rgate$-gate gadgets for the computation run as given in \expref{Sections}{sec:H-gadget} and~\ref{sec:R-gadget}.
\item For each gate $G$ in  $\cC$:  $\xgate, \zgate$ and $\cnot$ are performed without any auxiliary qubits or interaction as given in \expref{Sections}{sec:Pauli-gadget} and~\ref{sec:CNOT-gadget}, while the $\hgate$- and $\rgate$-gate gadgets are performed using the auxiliary qubits from \stepref{step:CR}{step:comp-aux} and the interaction as given in \expref{Sections}{sec:H-gadget} and~\ref{sec:R-gadget}, respectively.
\item $P$ measures the output qubit and returns the result to $V$.
\item $V$ decrypts the answer; let the result be~$c_{\text{comp}}$. $V$ accepts if~$c_{\text{comp}}=0$; otherwise reject.
\end{enumerate}

\item   \textbf{$\xgate$-test Run}\label{step:XR}
\begin{enumerate}[label*=\arabic*.,ref=\arabic*]
\item The verifier encrypts input $\ket{0}^{\otimes n}$  %
and sends~the input  to $P$.
\item  \label{step:comp-aux-x-test}The verifier sends auxiliary qubits required for the $\rgate$-gate gadgets for the $\xgate$-test run as given in
\expref{Sections}{sec:H-gadget} and~\ref{sec:R-gadget}.
\item  \label{step:comp-x-test} For each gate $G$ in  $\cC$:  $\xgate, \zgate$ and $\cnot$ are performed without any auxiliary qubits or interaction as given in \expref{Sections}{sec:Pauli-gadget} and~\ref{sec:CNOT-gadget}, while the $\hgate$- and $\rgate$-gate gadgets are performed using the auxiliary qubits from \stepref{step:XR}{step:comp-aux-x-test} and the interaction as given in \expref{Sections}{sec:H-gadget} and~\ref{sec:R-gadget}, respectively.
\item $P$ measures the output qubit and returns the result to $V$.
\item $V$ decrypts the answer; let the result be~$c_{\text{test}}$. $V$ accepts if no errors were detected in \stepref{step:XR}{step:comp-x-test} \emph{and} if $c_{\text{test}}=0$; otherwise reject.
\end{enumerate}

\item   \textbf{$\zgate$-test Run}\label{step:ZR}
\begin{enumerate}[label*=\arabic*.,ref=\arabic*]
\item The verifier encrypts input $\ket{+}^{\otimes n}$  %
and sends~the input  to $P$.
\item  \label{step:comp-aux-z-test}The verifier sends auxiliary qubits required for the $\rgate$-gate gadgets for the $\zgate$-test run as given in
\expref{Sections}{sec:H-gadget} and~\ref{sec:R-gadget}.
\item \label{step:comp-z-test} For each gate $G$ in  $\cC$:  $\xgate, \zgate$ and $\cnot$ are performed without any auxiliary qubits or interaction as given in \expref{Sections}{sec:Pauli-gadget} and~\ref{sec:CNOT-gadget}, while the $\hgate$- and $\rgate$-gate gadgets are performed using the auxiliary qubits from \stepref{step:ZR}{step:comp-aux-z-test} and the interaction as given in \expref{Sections}{sec:H-gadget} and~\ref{sec:R-gadget}, respectively.
    \item $P$ measures the output qubit and returns the result to $V$.
\item $V$ disregards the output. $V$ accepts if no errors were detected in \stepref{step:ZR}{step:comp-z-test}; otherwise reject.
\end{enumerate}
\end{enumerate}
\end{enumerate}
\end{proof-system}

\subsection{\texorpdfstring{$\xgate$}{X}- and \texorpdfstring{$\zgate$}{Z}-gate gadget}
\label{sec:Pauli-gadget}
In order to apply an $\xgate$ on a qubit register~$i$ encrypted with key $(a_i,b_i)$, the verifier updates the key according to $a_i \leftarrow a_i \oplus 1$ ($b_i$ is unchanged). In order to apply an $\zgate$ on a qubit register~$i$ encrypted with key $(a_i,b_i)$, the verifier updates the key according to $b_i \leftarrow b_i \oplus 1$ ($a_i$ is unchanged).
This operation is performed only in the computation run.

\subsection{\texorpdfstring{$\cnot$}{CNOT}-gate gadget}
\label{sec:CNOT-gadget}
In order to apply a $\cnot$ gate on the encrypted registers (say with register $i$ being the control and register $j$ the target), the prover simply applies the $\cnot$ gate on the respective registers. The verifier updates the encryption keys according to
$a_i  \leftarrow a_i; b_i  \leftarrow b_i \oplus b_j;
a_j  \leftarrow a_i \oplus a_j; \text{ and }
b_j  \leftarrow b_j$.
We mention that $\cnot (\ket{0}\ket{0}) = \ket{0}\ket{0}$ and $\cnot (\ket{+}\ket{+})=\ket{+}\ket{+}$; thus in the $\xgate$- and $\zgate$-test runs, the underlying data is unchanged.

\subsection{\texorpdfstring{$\rgate$}{T}-gate gadget}
\label{sec:R-gadget}

Here, we show how the $\rgate$ is performed on encrypted data. This is accomplished using an auxiliary qubit, as well as classical interaction.
 For the computation run, we use a combination of a protocol inspired from~\cite{FBS+14,Bro15}, as well as fault-tolerant quantum computation~\cite{BMP+00} (see also~\cite{BJ15}). This is given in \expref{Figure}{fig:R-gate-computation}. In the case of an $\xgate$ and $\zgate$ test runs, as usual, we want the identity map to be applied. This is done as in \expref{Figures}{fig:R-gate-X-test} and~\ref{fig:R-gate-Z-test}, respectively. Correctness of \expref{Figures}{fig:R-gate-computation}--\ref{fig:R-gate-Z-test} is proven in \expref{Section}{sec:correctness-R-gate}. Note that we show in \expref{Section}{section:resources-T-gate} that the set of auxiliary quantum states required by the verifier can be reduced via a simple re-labeling, in order to match the resources required in \expref{Theorem}{thm:main-QPIP}.

Also, note that in this work, we have slightly sacrificed efficiency for clarity in the proof, namely that we could have defined a $\pgate$-gadget using one simple auxiliary qubit instead of two auxiliary qubits that are used by implementing  the $\pgate$ as $\rgate^2$. Furthermore, we suspect that the $P^y$ gate is unnecessary in \expref{Figure}{fig:R-gate-computation} and thus that we can simplify the set $\mathcal{S}$ (however, the proof in this case appears to be more elaborate, so once more we choose clarity of the proof over efficiency).

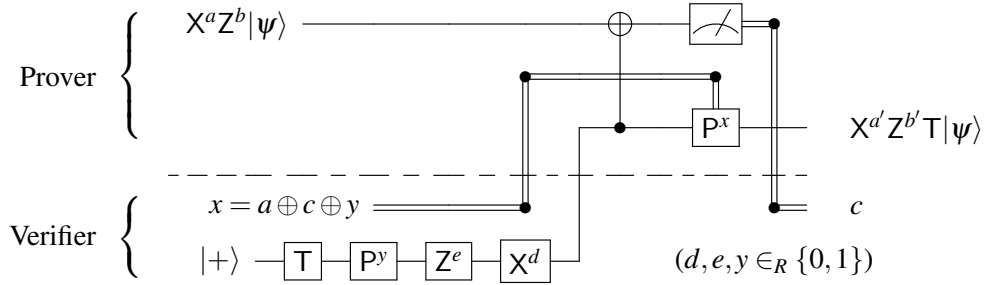
\begin{figure}[H]
\centerline{ %
 \Qcircuit @C=1em @R=1em  {
&&&&&&\lstick{\xgate^a \zgate^b\ket{\psi}} & \qw & \qw & \qw & \qw & \targ &\qw    &  \meter & \control \cw \cwx[4]  &    &   \\
\lstick{\text{Prover}}&&& & && &&& \control  & \cw    & \cw
&   \cw & \control \cw  & &&\\
&&&& & & && &&& \ctrl{-2} &\qw    & \gate{\pgate^x} \cwx &  \qw  &
\qw& \rstick{\xgate^{a'}  \zgate^{b'}  \tgate \ket{\psi}}\\  %
&&&\dw&\dw&\dw&\dw& \dw & \dw & \dw &\dw & \dw
 &  \dw & \dw  & \dw &\dw&\\
\lstick{\raisebox{-1cm}{\text{Verifier}}}&& &&&& &\lstick{x = a\oplus c \oplus y } &\cw & \control \cw \cwx[-3] & &
&   & & \control &  \cw& \rstick{c} \\
&&&&& \lstick{\ket{+}}   & \gate{\rgate} & \gate{\pgate^y} & \gate{\zgate^{e}} &\gate{\xgate^d}&
\qw \qwx[-3]
 &   & & &\mbox{($d, e, y \in_R \{0,1\} $)} & &&
 \rstick{}
 \gategroup{1}{2}{3}{2}{0.7em}{\{}
 \gategroup{5}{2}{6}{2}{0.7em}{\{}
 }
 }
 \caption{\label{fig:R-gate-computation} $\rgate$-gate gadget for a computation run. Here, an auxiliary qubit is prepared by the verifier in the state $\xgate^d\zgate^e\pgate^y\tgate\ket{+}$ and sent to the prover. The prover performs a $\cnot$ between the auxiliary register and the data register; and then measures the data register. Given the measurement result,~$c$, the verifier sends a classical message, $x= a \oplus c \oplus y $ to the prover, who applies the conditional gate $\pgate^x$ to the remaining register (which we now re-label as the data register). The verifier's key update rule is given  by $a' =a\oplus c $ and $b'= (a \oplus c)\cdot (d \oplus y)   \oplus a \oplus b \oplus c \oplus e  \oplus y$\,.}%
\end{figure}

\begin{figure}[H]
\centerline{
 \Qcircuit @C=1em @R=1em  {
&&&&&&\lstick{\xgate^a\ket{0}} & \qw & \qw & \qw & \qw & \targ &\qw    &  \meter & \control \cw \cwx[4]  &    &   \\
\lstick{\text{Prover}}&&& & && &&& \control  & \cw    & \cw
&   \cw & \control \cw  & &&\\
&&&& & & && &&& \ctrl{-2} &\qw    & \gate{\pgate^x} \cwx &  \qw  &
\qw& \rstick{\xgate^d \ket{0}}\\
&&&\dw&\dw&\dw&\dw& \dw & \dw & \dw &\dw & \dw
 &  \dw & \dw  & \dw &\dw&\\
\lstick{\raisebox{-1cm}{\text{Verifier}}}&& &&&& &\lstick{x \in_R \{0,1\}} &\cw & \control \cw \cwx[-3] & &
&   & & \control &  \cw& \rstick{c = a \oplus d} \\
&&&&& &\lstick{\ket{0}}   & \gate{\xgate^d} & \qw & \qw &
\qw \qwx[-3]
 &   & & &\mbox{($d \in_R \{0,1\} $)} & &&
 \rstick{}
 \gategroup{1}{2}{3}{2}{0.7em}{\{}
 \gategroup{5}{2}{6}{2}{0.7em}{\{}
 }
 }
 \caption{\label{fig:R-gate-X-test} $\rgate$-gate gadget for an $\xgate$-gate test run. The goal here is to mimic the interaction established in \expref{Figure}{fig:R-gate-computation}, but to perform the identity operation on the input state~$\ket{0}$ (up to encryptions). Here, we include an additional \emph{verification} that~$c=a \oplus d$.}
\end{figure}
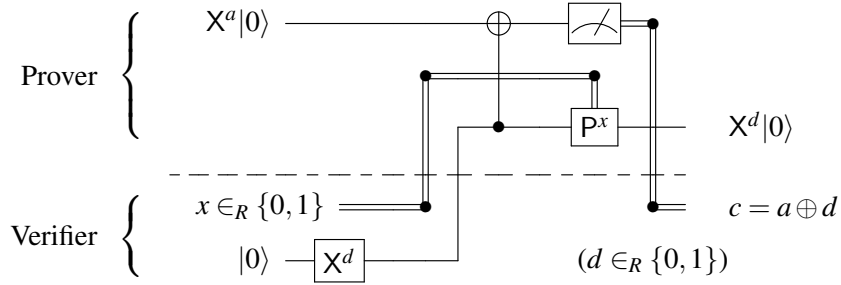

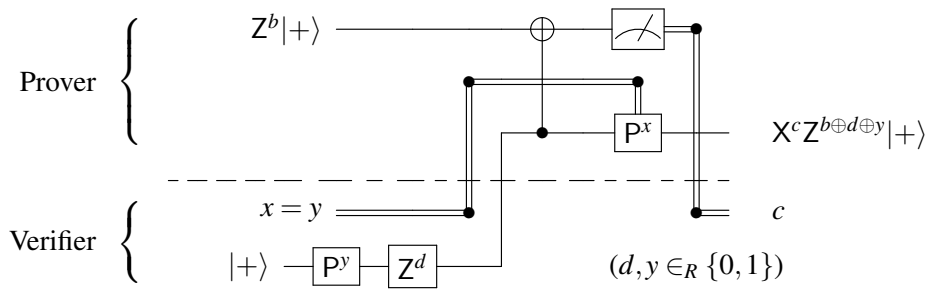
\begin{figure}[H]
\centerline{
 \Qcircuit @C=1em @R=1em  {
&&&&&&&\lstick{\zgate^b\ket{+}} & \qw & \qw & \qw & \targ &\qw    &  \meter & \control \cw \cwx[4]  &    &   \\
\lstick{\text{Prover}}&&& & && &&& \control  & \cw    & \cw
&   \cw & \control \cw  & &&\\
&&&& & & && &&& \ctrl{-2} &\qw    & \gate{\pgate^x} \cwx &  \qw  &
\qw& \rstick{ \xgate^c \zgate^{b\oplus d \oplus y} \ket{+}}\\
&&&\dw&\dw&\dw&\dw& \dw & \dw & \dw &\dw & \dw
 &  \dw & \dw  & \dw &\dw&\\
\lstick{\raisebox{-1cm}{\text{Verifier}}}&& &&&& &\lstick{x =y } &\cw & \control \cw \cwx[-3] & &
&   & & \control &  \cw& \rstick{c} \\
&&&&& &\lstick{\ket{+}}   & \gate{\pgate^y} & \gate{\zgate^{d}} & \qw &
\qw \qwx[-3]
 &   & & &\mbox{($d, y \in_R \{0,1\} $)} & &&
 \rstick{}
 \gategroup{1}{2}{3}{2}{0.7em}{\{}
 \gategroup{5}{2}{6}{2}{0.7em}{\{}
 }
 }
 \caption{\label{fig:R-gate-Z-test} $\rgate$-gate gadget for a  $\zgate$-gate test run.  The goal here is to mimic the interaction established in \expref{Figure}{fig:R-gate-computation}, but to perform the identity operation on the input state~$\ket{+}$ (up to encryptions).}
\end{figure}

\subsection{\texorpdfstring{$\hgate$}{H}-gate gadget}
\label{sec:H-gadget}

Performing a $\hgate$ gate has the effect of locally swapping between the $\xgate$- and $\zgate$-test runs (as well as swapping the role of the $\xgate$ and $\zgate$ encryption keys in the computation run).  While this is alright if done in isolation, it does not work if the $\hgate$-gate is performed as part of a larger computation (for instance, a $\cnot$-gate could no longer be performed as given above as the inputs would not, in general, be of the form $\ket{0}\ket{0}$ (for the $\xgate$-test run) or $\ket{+}\ket{+}$ (for the $\zgate$-test run)). Our solution is to use the following two
identities. %
\begin{align} \label{eqn:HPHPHPH}
 \hgate \pgate \hgate \pgate \hgate \pgate \hgate = \hgate\,,\\
 \hgate \hgate \hgate \hgate = \mathbb{I}\,.
 \end{align}
Thus we build the gadget so that the prover starts by applying an $\hgate$ at the start. By doing this,  we locally swap the roles of the $\xgate$- and $\zgate$-tests we also cause a key update which swaps the role of the $\xgate$- and $\zgate$-encryption keys. For the following $\pgate$, we apply twice the gadgets from \expref{Section}{sec:R-gadget}  (taking in to account the swapped role for the test runs). The result is that a $\pgate$ is applied in the computation run, while the identity is applied in the test runs. Now an $\hgate$ is applied, which reverts the roles of the $\xgate$- and $\zgate$-tests. We apply the $\pgate$ again. Continuing in this fashion,
we observe the following effect.  %
\begin{enumerate}
\item In the computation run (using twice the gadget of \expref{Figure}{fig:R-gate-computation} for each $\pgate$-gate), the effect is to apply~$\hgate$ on the input qubit (by \expeqref{Equation}{eqn:HPHPHPH}).
\item In the $\xgate$-test run (using (twice each time) the gadgets of \expref{Figures}{fig:R-gate-Z-test},   \ref{fig:R-gate-X-test},  \ref{fig:R-gate-Z-test} for the first, second and third $\pgate$-gate), the effect is to apply the identity.
\item In the $\zgate$-test run (using (twice each time) gadgets of \expref{Figures}{fig:R-gate-X-test},   \ref{fig:R-gate-Z-test},  \ref{fig:R-gate-X-test} for the first, second and third $\pgate$-gate), the effect is to apply the identity.
\end{enumerate}

\section{Correctness of the \texorpdfstring{$\rgate$}{T}-gate protocol}
\label{sec:correctness-R-gate}

\normalfont
We give below a step-by-step proof of the correctness of the
$\rgate$-gate protocol as given in \expref{Figure}{fig:R-gate-computation} (\expref{Section}{sec:R-gadget}).
In  \expref{Section}{sec:correctness-test}, we show how similar techniques are used to show corrections of the $\rgate$-gate protocol for the test runs, as given in \expref{Figures}{fig:R-gate-X-test} and~\ref{fig:R-gate-Z-test}. The basic
building block is the circuit identity for an \xgate-teleportation
from~\cite{ZLC00}. Also of relevance to
this work are the techniques developed by Childs, Leung, and
Nielsen~\cite{CLN05} to manipulate circuits that produce an output
that is correct \emph{up to known Pauli corrections}.

We will make use of the following identities which all hold up to an
irrelevant global phase: $ \xgate \zgate = \zgate \xgate$, $\pgate
\zgate = \zgate \pgate$, $\pgate \xgate = \xgate \zgate \pgate$,
$\rgate \zgate = \zgate \rgate$, $\rgate \xgate = \xgate \zgate
\pgate \rgate$,  $\pgate^2  = \zgate$ and $\pgate^{a\oplus b} =
\zgate^{a\cdot b} \pgate^{a + b}$ (for $a, b \in \{0,1\}$).

\begin{enumerate}

\item We start with the ``\xgate-teleportation''
of~\cite{ZLC00}, which is easy to verify (\expref{Figure}{fig:circuit-proof:2}).

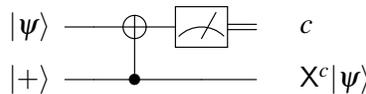
\begin{figure}[!ht]
 \centerline{
 \Qcircuit @C=1em @R=1em {
\lstick{\ket{\psi}} & \qw &\targ      &  \meter & \cw   &  \rstick{c}      \\
\lstick{\ket{+}}    & \qw &\ctrl{-1}  &   \qw  & \qw &
\rstick{\xgate^c\ket{\psi}}&   }
 }
 \caption{Circuit identity: ``$\xgate$-teleportation.''}
 \label{fig:circuit-proof:2}
\end{figure}

\item Then we substitute the input $\xgate^a \zgate^b \ket{\psi}$ for the top wire. We add the gate sequence $\tgate, \pgate^y, \zgate^e, \xgate^d$, $\pgate^{a \oplus c \oplus y }$ to the output
(\expref{Figure}{fig:circuit-proof:3}). By \expref{Figure}{fig:circuit-proof:2}, the outcome is given by $\pgate^{a \oplus c \oplus y } \xgate^d  \zgate^e \pgate^y \tgate \xgate^{a\oplus c}\zgate^b \ket{\psi}$.
We apply the identities given above to simplify this to a Pauli correction (on $\rgate$) %
as follows.

\begin{align} \label{eqn:circuit-3}
\pgate^{a \oplus c \oplus y} \xgate^d  \zgate^e \pgate^y \tgate \xgate^{a\oplus c}\zgate^b
& = \pgate^{a \oplus c \oplus y} \xgate^d  \zgate^e \pgate^y \xgate^{a\oplus c } \pgate^{a\oplus c} \zgate^{b \oplus a\oplus c} \tgate  \\
& = \pgate^{a \oplus c \oplus y } \xgate^d  \zgate^e \pgate^y  \pgate^{a\oplus c }\xgate^{a\oplus c }  \zgate^{b} \tgate  \\
& = \pgate^{a \oplus c \oplus y }  \pgate^y \xgate^d  \zgate^{d\cdot y \oplus e} \pgate^{a\oplus c}\xgate^{a\oplus c }  \zgate^{b} \tgate  \\
& = \pgate^{a \oplus c \oplus y }  \pgate^y  \pgate^{a\oplus c }\xgate^d  \zgate^{d\cdot(a \oplus c)} \zgate^{d\cdot y \oplus e}  \xgate^{a\oplus c}  \zgate^{b} \tgate  \\
& = \pgate^{(a \oplus c ) \oplus y}  \pgate^y  \pgate^{a\oplus c} \xgate^{a\oplus c \oplus d}  \zgate^{ d (a\oplus c \oplus y) \oplus b \oplus e } \tgate  \\
& = \zgate^{y\cdot(a \oplus c) } \pgate^{a \oplus c} \pgate^y \pgate^y  \pgate^{a\oplus c } \xgate^{a\oplus c \oplus d}  \zgate^{ d (a\oplus c \oplus y) \oplus b \oplus e } \tgate  \\
& =  \xgate^{a\oplus c \oplus d}  \zgate^{(a \oplus c \oplus d)\cdot (d \oplus y)   \oplus a \oplus b \oplus c \oplus d \oplus e  \oplus y   } \tgate \\
&= \xgate^{a\oplus c'}  \zgate^{(a \oplus c')\cdot (d \oplus y)   \oplus a \oplus b \oplus c' \oplus e  \oplus y}\tgate\,, \label{eqn:circuit-35}
\end{align}
where above we let $c' \leftarrow c \oplus d$.

\begin{figure}[!ht]
 \centerline{\hspace*{-2cm}
 \Qcircuit @C=1em @R=1em {
\lstick{\xgate^a \zgate^b \ket{\psi}} & \qw &\targ      &  \qw & \meter & \cw   &  \rstick{c; c' \leftarrow c\oplus d}      \\
\lstick{\ket{+}}    & \qw &\ctrl{-1}  &   \gate{\rgate} & \gate{\pgate^y} & \gate{ \zgate^e}  & \gate{\xgate^d} & \gate{\pgate^{a \oplus c \oplus y }} & \qw &
\rstick{ \xgate^{a\oplus c'}  \zgate^{(a \oplus c')\cdot (d \oplus y)   \oplus a \oplus b \oplus c' \oplus e  \oplus y}  \tgate \ket{\psi}}&   }
 }
 \caption{}
 \label{fig:circuit-proof:3}
\end{figure}

\item Next, we note that, because diagonal gates commute with control,  the circuit of \expref{Figure}{fig:circuit-proof:3} is equivalent to the one in \expref{Figure}{fig:circuit-proof:4}.

\begin{figure}[!ht]
 \centerline{
 \Qcircuit @C=1em @R=1em {
\lstick{\xgate^a \zgate^b \ket{\psi}} & \qw &\targ      &  \qw & \meter & \cw   & \rstick{c; c' \leftarrow c\oplus d}  \\
\lstick{\zgate^e\pgate^y\rgate\ket{+}}    & \qw &\ctrl{-1}  &   \gate{\xgate^d} & \gate{\pgate^{a \oplus c \oplus y}} & \qw &
\rstick{ \xgate^{a\oplus c'}  \zgate^{(a \oplus c')\cdot (d \oplus y)   \oplus a \oplus b \oplus c' \oplus e  \oplus y}  \tgate \ket{\psi}}   &   }
 }
 \caption{}
 \label{fig:circuit-proof:4}
\end{figure}

\item We note that the $\xgate^d$ on the bottom wire \emph{after} the $\cnot$ can be moved to the bottom wire  \emph{before} the $\cnot$, as long as we add an $\xgate^d$ to the top wire after the $\cnot$.  (\expref{Figure}{fig:circuit-proof:5}.) %

\begin{figure}[!ht]
 \centerline{
 \Qcircuit @C=1em @R=1em {
\lstick{\xgate^a \zgate^b \ket{\psi}} & \qw &\targ      &  \gate{\xgate^d} & \meter & \cw   &  \rstick{c; c' \leftarrow c\oplus d}      \\
\lstick{\xgate^d\zgate^e\pgate^y\rgate\ket{+}}    & \qw &\ctrl{-1}  &   \qw & \gate{\pgate^{a \oplus c \oplus y }} & \qw &
\rstick{ \xgate^{a\oplus c'}  \zgate^{(a \oplus c')\cdot (d \oplus y)   \oplus a \oplus b \oplus c' \oplus e  \oplus y}  \tgate \ket{\psi}} &   }
 }
 \caption{}
 \label{fig:circuit-proof:5}
\end{figure}

\item Finally, since $c' = c \oplus d$, yet the measurement result $c$ undergoes an $\xgate^d$, these two operations cancel out, and we obtain the final circuit as in \expref{Figure}{fig:circuit-proof:6}.
\begin{figure}[!ht]
 \centerline{
 \Qcircuit @C=1em @R=1em {
\lstick{\xgate^a \zgate^b \ket{\psi}} & \qw &\targ      &  \qw & \meter & \cw   &  \rstick{c}      \\
\lstick{\xgate^d\zgate^e\pgate^y\rgate\ket{+}}    & \qw &\ctrl{-1}  &   \qw & \gate{\pgate^{a \oplus c \oplus y }} & \qw &
\rstick{ \xgate^{a\oplus c}  \zgate^{(a \oplus c)\cdot (d \oplus y)   \oplus a \oplus b \oplus c \oplus e  \oplus y}  \tgate \ket{\psi}} &   }
 }
 \caption{}
 \label{fig:circuit-proof:6}
\end{figure}
\end{enumerate}

We note that a more direct proof of correctness for \expref{Figure}{fig:circuit-proof:6} is possible, but that our intermediate \expref{Figure}{fig:circuit-proof:5} is crucial in the proof of soundness.

\subsection{Correctness of the \texorpdfstring{$\tgate$}{T}-gate gadget in the test runs}
\label{sec:correctness-test}

The correctness of the $\tgate$-gate gadget in the $\xgate$-test run  of \expref{Figure}{fig:R-gate-X-test} is straightforward: the $\cnot$ flips the bit in the top wire if and only if $d=1$, while the $\pgate^x$ has no effect on the computational basis states.
The correctness of \expref{Figure}{fig:R-gate-Z-test} is derived from the $\xgate$-teleportation of \expref{Figure}{fig:circuit-proof:2}. Since the diagonal gates $\zgate$ and $\pgate$ commute with control,  they can be seen as acting on the output qubit. Furthermore,  using $\pgate^2 = \zgate$ and the fact that $\xgate$ has no effect on $\ket{+}$, we get the final circuit in \expref{Figure}{fig:R-gate-Z-test}.

\section{Completeness}
\label{sec:completeness}

Suppose $\mathcal{C}$ is a yes-instance of \textsf{Q-CIRCUIT}. Suppose $P$ follows the protocol honestly. 
Then we have the following.
\begin{enumerate}
\item In the case of a computation run, the output bit, $c_\text{comp}$ has the same distribution as the output bit of $C(\ket{0^n})$, thus  $V$ accepts with probability at least ${2}/{3}$.
\item In the case of an $\xgate$-test run and in the case of a $\zgate$-test run (by the identities and observations from the previous sections), $V$ accepts with probability~$1$.
\end{enumerate}

Given that each run happens with probability ${1}/{3}$, we get that $V$ accepts with probability at least ${2}/{3} + ({1}/{3}) \cdot ({2}/{3}) = {8}/{9}$.

\subsection{Auxiliary qubits for the \texorpdfstring{$\tgate$}{T}-gate gadget}
\label{section:resources-T-gate}
In the protocol for the $\tgate$-gate gadget (\expref{Figure}{fig:R-gate-computation}), we assume the verifier can produce auxiliary qubits of the form $\xgate^d\zgate^e\pgate^y\tgate\ket{+}$. We now show that this is equivalent to requiring the verifier to generate auxiliary qubits of the form $\zgate^e \pgate^y \tgate \ket{+}$, as claimed in \expref{Theorem}{thm:main-QPIP}.
This can be seen by the following equation, which holds up to global phase. 
\begin{equation}
\label{eqn:relabelling}
\xgate^d\zgate^e \pgate^y \tgate \ket{+} = \zgate^{e  \oplus d} \pgate^{y \oplus d} \tgate \ket{+}\,.
\end{equation}
The above can be seen easily since, up to global phase, $\xgate \rgate \ket{+} = \zgate \pgate \rgate \ket{+} $, and $\xgate \pgate = \zgate \pgate \xgate$%
.
The analysis for the case $d=1$, follows. 
\begin{align}
\xgate\zgate^e \pgate^y \tgate \ket{+} &= \zgate^e \xgate  \pgate^y \tgate \ket{+}\\
& = \zgate^e \zgate^y \pgate^y \xgate \tgate \ket{+} \\
&= \zgate^{e\oplus y}  \pgate^y \zgate \pgate  \tgate \ket{+} \\
&= \zgate^{e\oplus y \oplus 1}  \pgate^{y +1} \tgate \ket{+}\\
&=  \zgate^{e\oplus y \oplus 1}  \zgate^y \pgate^{y \oplus 1} \tgate \ket{+}\\
&=  \zgate^{e \oplus 1} \pgate^{y \oplus 1} \tgate \ket{+}\,.
\end{align}
Thus, the verifier chooses a classical $x$ uniformly at random, and if $x=1$, the verifier re-labels the auxiliary qubits according to \expeqref{Equation} {eqn:relabelling}.

\section{Soundness}
\label{sec:soundness}

As discussed in \expref{Section}{sec:overview-techniques}, the main idea to prove soundness is to analyze an entanglement-based version of the \expref{Interactive Proof System}{prot:interactive-proof}. We present the EPR-based version (\expref{Section}{sec:EPR-based-QPIP}),  and argue why the completeness and soundness parameters are the same.
Then, we analyze a general deviating prover $P^*$ in the EPR-based version and show how to simplify an attack (\expref{Section}{sec:simplifying-general}). We then analyze the case of a test run (\expref{Section}{sec:soundness-testrun}) and of a computation run (\expref{Sections}{sec:soundness-computation run}). In \expref{Section}{sec:soundness-proof}, we show how this completes the proof of our main theorem (\expref{Theorem}{thm:main-QPIP}).

An interesting consequence of the analysis in this section is that it implies that,
if we are willing to have the prover and the verifier share entanglement,
then the protocol reduces to a single round.
(However, in this case, the work of the verifier becomes
more important; one can wonder if the verifier is still ``almost-classical''.)
Another interesting observation is
that sequential repetition is not required (parallel repetition suffices), due to the fact that the analysis  makes use of the Pauli twirl (see \expref{Section}{sec:simplifying-general}), which would also be applicable to the scenario of parallel repetition.

\subsection{EPR-based protocol}
\label{sec:EPR-based-QPIP}

In this version of the quantum-prover interactive proof system (\textbf{\expref{Interactive Proof System}{prot:interactive-proof-EPR}}), \emph{all} quantum inputs sent by the verifier are half-EPR pairs, and \emph{all} classical messages sent by the verifier are random bits. The actions related to choosing between test and computation runs are done \emph{after} the interaction with the server.  For the~$\rgate$-gate, this can be done as shown in \expref{Figures}{fig:R-gate-EPR-computation}, \ref{fig:R-gate-EPR-X-test} and~\ref{fig:R-gate-EPR-Z-test}.

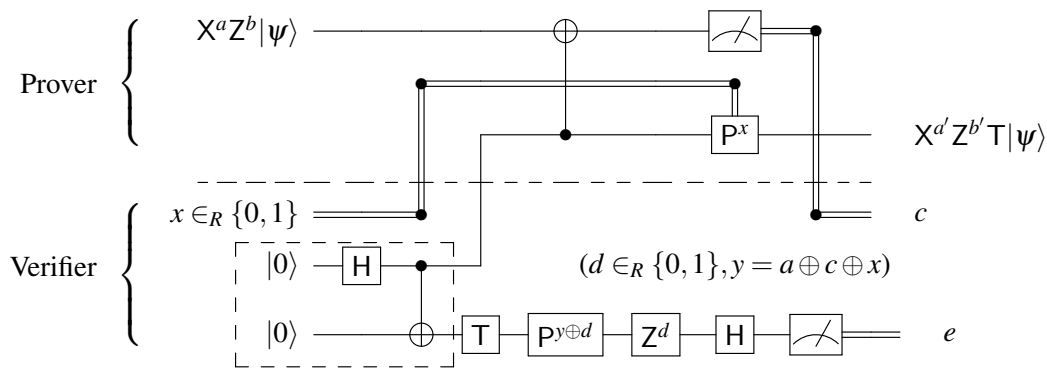
\begin{figure}[H]
\centerline{%
 \Qcircuit @C=1em @R=1em  {
&&&&&&&\lstick{\xgate^a \zgate^b\ket{\psi}} & \qw & \qw & \qw & \targ &\qw    &  \meter & \control \cw \cwx[4]  &    &   \\
\lstick{\text{Prover}}&&&&& & && & \control  & \cw    & \cw
&   \cw & \control \cw  & &&\\
&&&&&& & & && & \ctrl{-2} &\qw    & \gate{\pgate^x} \cwx &  \qw  &
\qw& \rstick{\xgate^{a'}  \zgate^{b'} \tgate   \ket{\psi}}\\ %
&&&&\dw &\dw&\dw& \dw & \dw & \dw &\dw & \dw
 &  \dw & \dw  & \dw &\dw&\\
&& &&&& &\lstick{x \in_R \{0,1\}} &\cw & \control \cw \cwx[-3] & &
&   & & \control &   \cw  & \rstick{c} \\
\lstick{\text{Verifier}}&&& &&&&\lstick{\ket{0}}    & \gate{\hgate}
&\ctrl{1} & \qw \qwx[-3]
 & &  &\mbox{($ d \in_R \{0,1\}, y = a \oplus c \oplus x $)}  & & &&\\
&& &&&& &\lstick{\ket{0}}    & \qw & \targ & \gate{\rgate}
 & \gate{\pgate^{y \oplus d}}  & \gate{\zgate^{d}}& \gate{\hgate} &  \meter &  \cw  & \cw
   & \rstick{e}
    \gategroup{6}{6}{7}{10}{1.4em}{--}
     \gategroup{1}{2}{3}{2}{0.7em}{\{}
 \gategroup{5}{2}{7}{2}{0.7em}{\{}
 }
 }
  \caption{\label{fig:R-gate-EPR-computation}Entanglement-based protocol for a $\rgate$-gate (computation run). This protocol
performs the same computation as the protocol in
\expref{Figure}{fig:R-gate-computation}. The output is obtained from  the output of \expref{Figure}{fig:R-gate-computation} by using $y = a \oplus c \oplus x$. The circuit in
the dashed box prepares an EPR-pair. Here, $a'= a\oplus c$ and $b'=(a \oplus c) \cdot (d \oplus y) \oplus a \oplus  b \oplus c \oplus e \oplus y$\,.} %
\end{figure}

\begin{figure}[H]
\centerline{
 \Qcircuit @C=1em @R=1em  {
&&&&&&&\lstick{\xgate^a\ket{0}} & \qw & \qw & \qw & \targ &\qw    &  \meter & \control \cw \cwx[4]  &    &   \\
\lstick{\text{Prover}}&&&&& & && & \control  & \cw    & \cw
&   \cw & \control \cw  & &&\\
&&&&&& & & && & \ctrl{-2} &\qw    & \gate{\pgate^x} \cwx &  \qw  &
\qw& \rstick{\xgate^{d} \ket{0}}\\
&&&&\dw &\dw&\dw& \dw & \dw & \dw &\dw & \dw
 &  \dw & \dw  & \dw &\dw&\\
&& &&&& &\lstick{x \in_R \{0,1\}} &\cw & \control \cw \cwx[-3] & &
&   & & \control &   \cw  & \rstick{c = a \oplus d} \\
\lstick{\text{Verifier}}&&& &&&&\lstick{\ket{0}}    & \gate{\hgate}
&\ctrl{1} & \qw \qwx[-3]
 &   & & & & &&\\
&& &&&& &\lstick{\ket{0}}    & \qw & \targ & \qw
 & \qw  & \qw &  \meter &  \cw  & \cw
   & \rstick{d}
    \gategroup{6}{6}{7}{10}{1.4em}{--}
     \gategroup{1}{2}{3}{2}{0.7em}{\{}
 \gategroup{5}{2}{7}{2}{0.7em}{\{}
 }
 }
  \caption{\label{fig:R-gate-EPR-X-test}Entanglement-based protocol for a $\rgate$-gate ($\xgate$-test run). This protocol
performs the same computation as the protocol in
\expref{Figure}{fig:R-gate-X-test}. The circuit in
the dashed box prepares an EPR-pair. As in \expref{Figure}{fig:R-gate-X-test}, we include an additional \emph{verification} that $c=a \oplus d$.}
\end{figure}
\begin{figure}[H]
\centerline{
 \Qcircuit @C=1em @R=1em  {
&&&&&&&\lstick{\zgate^b\ket{+}} & \qw & \qw & \qw & \targ &\qw    &  \meter & \control \cw \cwx[4]  &    &   \\
\lstick{\text{Prover}}&&&&& & && & \control  & \cw    & \cw
&   \cw & \control \cw  & &&\\
&&&&&& & & && & \ctrl{-2} &\qw    & \gate{\pgate^x} \cwx &  \qw  &
\qw& \rstick{\xgate^{c}\zgate^{b \oplus d \oplus y}  \ket{+}}\\
&&&&\dw &\dw&\dw& \dw & \dw & \dw &\dw & \dw
 &  \dw & \dw  & \dw &\dw&\\
&& &&&& &\lstick{x \in_R \{0,1\}} &\cw & \control \cw \cwx[-3] & &
&   & & \control &   \cw  & \rstick{c} \\
\lstick{\text{Verifier}}&&& &&&&\lstick{\ket{0}}    & \gate{\hgate}
&\ctrl{1} & \qw \qwx[-3]
 &   & \mbox{($y = x $)}& & & &&\\
&& &&&& &\lstick{\ket{0}}    & \qw & \targ & \qw
 & \gate{\pgate^y} &  \gate{\hgate} &  \meter   & \cw
   & \rstick{d }
    \gategroup{6}{6}{7}{10}{1.4em}{--}
     \gategroup{1}{2}{3}{2}{0.7em}{\{}
 \gategroup{5}{2}{7}{2}{0.7em}{\{}
 }
 }
  \caption{\label{fig:R-gate-EPR-Z-test}Entanglement-based protocol for a $\rgate$-gate ($\zgate$-test run). This protocol
performs the same computation as the protocol in
\expref{Figure}{fig:R-gate-Z-test}. The circuit in
the dashed box prepares an EPR-pair.}
\end{figure}

We therefore define $V_{\text{EPR}}$ as a verifier that delays all choices until after the Prover has returned all messages (\ie, as the verifier in \expref{Interactive Proof System}{prot:interactive-proof-EPR}).
That the soundness parameter is the same for \expref{Interactive Proof Systems}{prot:interactive-proof} and~\ref{prot:interactive-proof-EPR} follows from the following series of transformations, each of which preserves the probability of accept/reject.
\begin{enumerate}
\item The  quantum communication from the verifier to the prover in \expref{Interactive Proof System}{prot:interactive-proof} can be generated instead by preparing EPR pairs, sending half and then immediately measuring in the appropriate basis to obtain the required qubit. Call this Interactive Proof System 1.1.
\item  Starting from Interactive Proof System 1.1, by a change of variable, we can specify the classical message $x$ to be random in each of the $\tgate$-gate gadgets. This determines~$y$, and the operation
 $\pgate^y$, that is conditioned on~$y$ (for the computation and $\zgate$-test run) can be applied immediately before the verifier's measurement. Call this Interactive Proof System 1.2.
\item In Interactive Proof System 1.2, the prover and verifier operate on different subsystems. 
Their operations therefore commute and we can specify that the Verifier delays his operations until \emph{after} the interaction with the Prover. The result of this transformation is \expref{Interactive Proof System}{prot:interactive-proof-EPR}.
\end{enumerate}

Since each transformation above preserves the soundness, and since the completeness parameter is unchanged,
from now on, we can focus on establishing the soundness parameter for \expref{Interactive Proof System}{prot:interactive-proof-EPR}.

\begin{proof-system} \caption{\label{prot:interactive-proof-EPR} Verifiable quantum computation with trusted auxiliary states- EPR version}
Let $\mathcal{C}$ be given as an $n$-qubit quantum circuit in the universal 
set of gates
 $\xgate, \zgate, \cnot, \hgate, \rgate$.
\begin{enumerate}
\item The verifier prepares $\ket{\Phi^+}^{\otimes n}$ and sends  half of each pair to the prover. These registers are identified with the \emph{input} registers.
\item \label{step:gates-EPR-protocol} For each auxiliary qubit required in the $\hgate$- and $\rgate$-gate gadgets, the verifier prepares  $\ket{\Phi^+}$ and sends half of each pair to the prover.
\item The prover executes the gate gadgets. The verifier records the classical communication and responds with random classical bits (when required).\label{step:gates-EPR-protocol-execute}
\item The prover returns a single bit of output,~$c$ to the verifier.
\item The verifier randomly chooses to execute one of the following three runs (but does not inform the prover of this choice).

\begin{enumerate}[label=\Alph*.]
\item   \textbf{Computation Run}
\begin{enumerate}[label*=\arabic*.]
\item \label{step:input-EPR-Computation}Measure the remaining input register halves in the computational basis. Take the initial $\xgate$-encryption key to be the measurement outcomes (set the $\zgate$-key to 0).
\item For each gate $G$ in  $\cC$:  perform the key updates for the $\xgate, \zgate$ , $\cnot$ and $\hgate$ gates. For the $\rgate$ gadget, taking into account the classical messages received and sent in \expref{Step}{step:gates-EPR-protocol-execute}, perform the measurement and key update rules for the $\tgate$-gadget (\expref{Figure}{fig:R-gate-EPR-computation}).
\item $V$ decrypts the output bit~$c$; let the result be~$c_{\text{comp}}$. $V$ accepts if~$c_{\text{comp}}=0$; otherwise reject.
\end{enumerate}

\item   \textbf{$\xgate$-test Run}\label{step:2XR}
\begin{enumerate}[label*=\arabic*.,ref=\arabic*]
\item \label{step:input-EPR-X-test} Measure the remaining input register halves in the computational basis. Take the initial $\xgate$-encryption key to be the measurement outcomes (set the $\zgate$-key to 0).
\item \label{step:comp-x-test-EPR} For each gate $G$ in  $\cC$:  perform the key updates for the $\xgate, \zgate$, $\cnot$ and $\hgate$ gates. For the $\rgate$ gadget, taking into account the classical messages received and sent in \expref{Step}{step:gates-EPR-protocol-execute}, perform the measurement, key update rules and tests for the $\tgate$-gadget (\expref{Figure}{fig:R-gate-EPR-X-test}).
\item $V$ decrypts the output bit~$c$; let the result be~$c_{\text{comp}}$. $V$ accepts if~$c_{\text{comp}}=0$ \emph{and} if no errors were detected in \stepref{step:2XR}{step:comp-x-test-EPR}; otherwise reject.
\end{enumerate}

\item   \textbf{$\zgate$-test Run}\label{step:2ZR}
\begin{enumerate}[label*=\arabic*.,ref=\arabic*]
\item \label{step:input-EPR-Z-test} Measure the remaining input register halves in the Hadamard basis. Take the initial $\zgate$-encryption key to be the measurement outcomes (set the $\xgate$-key to 0).
\item \label{step:comp-z-test-EPR} For each gate $G$ in  $\cC$:  perform the key updates for the $\xgate, \zgate$, $\cnot$ and $\hgate$ gates. For the $\rgate$ gadget, taking into account the classical messages received and sent in \expref{Step}{step:gates-EPR-protocol-execute}, perform the measurement, key update rules and tests for the $\tgate$-gadget (\expref{Figure}{fig:R-gate-EPR-Z-test})
\item  $V$ accepts if no errors were detected in \stepref{step:2ZR}{step:comp-z-test-EPR}; otherwise reject.
\end{enumerate}
\end{enumerate}
\end{enumerate}

\end{proof-system}

\subsection{Simplifying a general attack}
\label{sec:simplifying-general}

In this section, we derive a simplified expression for a general deviating prover for our Interactive Proof System. %
We  
show  that without loss of generality, we can
rewrite  %
the actions of any deviating prover as the honest prover's actions, followed by an arbitrary cheating map. But first, as a technicality, we consider Interactive Proof System~3, which is closely related to \expref{Interactive Proof System}{prot:interactive-proof-EPR}, but  where the prover is unitary and show (\expref{Lemma}{lem:purified-EPR-protocol}) that a bound on the soundness of  Interactive Proof System~3 implies a bound on the soundness of \expref{Interactive Proof System}{prot:interactive-proof-EPR}.

\begin{definition}
  We define \emph{Interactive Proof System~3} as \expref{Interactive Proof System}{prot:interactive-proof-EPR}, but with a \emph{unitary} prover. To be more precise, the prover (say, $P^q$) in Interactive Proof System~3 performs the same operations as the prover in \expref{Interactive Proof System}{prot:interactive-proof-EPR}, but does not perform any measurements:  instead, $P^q$ sends qubits to the verifier, say  $V^{q}_{\text{EPR}}$, who immediately measures them in place of the prover, and then continues according to the \expref{Interactive Proof System}{prot:interactive-proof-EPR}.
\end{definition}

\begin{lemma} \label{lem:purified-EPR-protocol}
Interactive Proof Systems~2 and~3 have the same completeness, and furthermore, an upper bound on the soundness parameter for Interactive Proof System~3 is an upper bound on the soundness parameter for Interactive Proof System~2.
\end{lemma}

\begin{proof}
It follows immediately by definition that Interactive Proof Systems~2 and~3 have the same completeness parameter.
For soundness,
suppose that in Interactive Proof System~3, the probability that $V^{q}_{\text{EPR}}$ accepts while interacting with \emph{any} ${P^q}^*$ is at most~$s$. Suppose for a 
contradiction
that there is a $P^*$ for \expref{Interactive Proof System}{prot:interactive-proof-EPR}, such that  $V_{\text{EPR}}$ accepts with $p > s$.

Using polynomial overhead, we  obtain $\widetilde{P^*_\text{meas}}$ from $P^*$ via purification (all actions of $\widetilde{P^*_\text{meas}}$ are unitary, except a one-qubit register is measured each time the protocol requires a classical message to be sent to the verifier). The probability $p$ of acceptance is unchanged.

Furthermore, starting with $P^*$, we define $\widetilde{P^*}$ as a prover for Interactive Proof System~3, which  behaves like $\widetilde{P^*_\text{meas}}$, but instead of measuring messages to be sent to the verifier, it sends qubits, which are immediately measured by $V^{q}_{\text{EPR}}$, as per Interactive Proof System~3.
The probability $p$ that the verifier in Interactive Proof System~3 accepts, when interacting with $\widetilde{P^*}$ is the same as the probability that the verifier in \expref{Interactive Proof System}{prot:interactive-proof-EPR} accepts, when interacting with $P^*$. This contradicts $p > s$, and proves the claim.
 \end{proof}

 Next, we show that, in Interactive Proof System~3, without loss of generality, we can assume that the prover's actions are the honest unitary ones, followed by a general
attack.
In order to see this, for a $t$-round protocol (involving $t$ rounds of classical interaction),  define a cheating prover's actions at round~$i$ %
by $\Phi_i H_i$, where $H_i$  acts on the qubits used in the computation, as well as the classical bit received in round~$i$, and is the honest application of the prover's unitary circuit, while $\Phi_i$ is a general deviating map acting on the  classical bit received in round $i$, the output registers of $H_i$ as well as a private memory register. (Recall that there are no measurements at this point---$V$ does the measurement.)

Thus the actions of a general prover~$P^*$ %
are given as %
the follows.  %
 \begin{equation}
 \label{eqn:general-attack}
\Phi_t  H_t\ldots\Phi_1   H_1 \Phi_0H_0\,.
 \end{equation}
 Since $H_0, \ldots ,H_t$ are unitary, we can
rewrite   %
 \expeqref{Equation} {eqn:general-attack}. 
\begin{align}
 \Phi_t  H_t\ldots  \Phi_3   H_3 \Phi_2   H_2 \Phi_1   H_1 \Phi_0H_0 &=
  \Phi_t  H_t\ldots  \Phi_3   H_3 \Phi_2   H_2 (\Phi_1   H_1 \Phi_0 {H_1}^*) H_1 H_0\\
&=    \Phi_t  H_t\ldots  \Phi_3   H_3 (\Phi_2   H_2 \Phi_1   H_1 \Phi_0 {H_1}^* {H_2}^*)H_2  H_1 H_0\\
&=    \Phi_t  H_t\ldots ( \Phi_3   H_3 \Phi_2   H_2 \Phi_1   H_1 \Phi_0 {H_1}^* {H_2}^* H_3^*)H_3 H_2  H_1 H_0\\
&=  (\Phi_t  H_t\ldots  \Phi_3   H_3 \Phi_2   H_2 \Phi_1   H_1 \Phi_0 {H_1}^* {H_2}^* H_3^* \ldots  H_t^*) H_t  \ldots  H_3   H_2 H_1 H_0\,.
\end{align}
Thus, by denoting a general attack by
\begin{equation}
  \Phi = \Phi_t  H_t\ldots  \Phi_3   H_3 \Phi_2   H_2 \Phi_1   H_1 \Phi_0 {H_1}^* {H_2}^* H_3^* \ldots  H_t^*\,,
\end{equation}
and the map corresponding to the honest prover as
$H = H_t \ldots H_3 H_2 H_1 H_0$, we get that without loss of
generality, we can assume that the prover's actions are the honest
ones, followed by a general attack:
\begin{equation}
 \label{eqn:general-attack-simple}
\Phi H\,.
\end{equation}

Taking \{$E_k$\} to be  Kraus terms associated with $\Phi$, and supposing a total of $m$ qubit registers are involved,
we get that the system after interaction with~$P^*$, where the initial state is
\begin{equation}
  \egoketbra{\Phi^+}^{\otimes m} =  \frac{1}{2^m}\sum_{i,j=0}^{2^m-1}\ketbra{ii}{jj}
\end{equation}
(here, we include the classical random bits, as they are uniformly
random and therefore we can represent them as maximally entangled
states), and where the verifier has not yet performed the gates and
measurements, can be described as%
\begin{equation}
\label{eqn:attack-on-EPR-I}
  \frac{1}{2^m}\sum_k  \sum_{i,j}  (I\otimes E_k H ) \ketbra{ii}{jj} (I\otimes H^* E_k^*)\,.
\end{equation}

For a fixed $k$, we write  $E_k$ and $E_k^*$ in the Pauli basis:
\begin{equation}
  E_k = \sum_Q \alpha_{k,Q} Q\qquad\text{and}\qquad E_k^* = \sum_{Q'} \alpha_{k,Q'}^* Q'\,.
\end{equation}
(To 
simplify notation, we assume throughout that $Q$, $Q'$ ranges
over~$\mathbb{P}_{m}$  
.)
By the completeness relation, we have:
\begin{equation}
  \sum_k\sum_Q\abs{\alpha_{k,Q}}^2 =1\,.
\end{equation}
When it is clear from context, we drop the ``$k$'' subscript, thus
denoting
\begin{equation}
  E_k = \sum_Q \alpha_{Q} Q\qquad\text{and}\qquad E_k^* = \sum_{Q'} \alpha_{Q'}^* Q'\,.
\end{equation}
In the following sections, we analyze the probability of acceptance,
as a function of the type of run and of the prover's attack. By
\expref{Lemma}{lem:purified-EPR-protocol}, a bound on the acceptance
probability gives a bound on the acceptance probability in
\expref{Interactive Proof System}{prot:interactive-proof-EPR} (which
is a bound on the acceptance probability of \expref{Interactive Proof
  System}{prot:interactive-proof}).

\subsection{Conventions and definitions}
\label{sec:conventions}

In addition to the convention of representing an attack as in \expeqref{Equation} {eqn:attack-on-EPR-I},
in the following \expref{Sections}{sec:soundness-testrun}--\ref{sec:soundness-computation run}, we use the following conventions.%
\begin{enumerate}
\item The circuit $C$ that we consider is already ``compiled'' in terms of the $\hgate$ gates as in \expref{Section}{sec:H-gadget} (the identity $\hgate = \hgate \pgate \hgate \pgate \hgate \pgate\hgate$ is already applied).
\item The number of $\tgate$-gate gadgets is $t$ (each such gadget uses two auxiliary qubits---one representing an auxiliary quantum bit, and one representing a classical bit~$x$), and the number of qubits in the computation is $n$. %
    Thus we have $m=2t+n$.
     \item In the $\tgate$-gate gadget, the auxiliary wire is swapped with the measured wire immediately before the measurement. %
     This
      way, we may picture that only auxiliary qubits are measured as part of the computation, and that the data registers for the input  represent the computation wires throughout.
         \item  Given the system as in \expeqref{Equation} {eqn:attack-on-EPR-I}, we suppose that the first $\tgate$-gate gadget uses the first EPR pair  as auxiliary quantum bit, and the second EPR pair as a qubit representing the classical bit (and so on for the following $\tgate$-gadgets). The last $n$ EPR pairs are the data qubits, and we suppose that at the end of the protocol, the last data qubit is the one that is measured, representing the output.

         \item Normalization constants are omitted when they are clear from context.
\end{enumerate}

Finally, we define \emph{benign} and \emph{non-benign} Pauli attacks, based on their effect on the protocol. %
As we will see,  benign attacks have no effect on the acceptance probability (because all qubits are either traced-out or measured in the computational basis). However, non-benign attacks may influence the acceptance probability.
\begin{definition}
   For a fixed Pauli $P \in \mathbb{P}_{m}$, we call it \emph{benign} if $P \in B_{t,n}$, where $B_{t,n}$ is the set of Paulis acting on $m=2t+n$ qubits, such that the \emph{measured} qubits in the protocol are acted on only by a gate in $\{I, \zgate\}$. Using the  above conventions, this means that $B_{t,n} = \{\{\{I, \zgate\}\otimes\mathbb{P}\}^{\otimes  t} \otimes \mathbb{P}_{n-1} \otimes   \{I, \zgate\}\}$. A Pauli~$P$ is called \emph{non-benign} if at least one \emph{measured} qubit in the protocol is acted on only by a gate in $\{\xgate, \ygate\}$. In analogy to the set of benign Paulis, we denote the set of non-benign Paulis acting on $m=2t+n$ qubits as~$B'_{t,n}$.
\end{definition}

\subsection{In the case of a test run}
\label{sec:soundness-testrun}
Based on the preliminaries of \expref{Section}{sec:simplifying-general} and~\expref{Section}{sec:conventions}, we now bound the probability of acceptance of the test runs, by describing the effect of the attack on the entire system, and considering which attacks are detected by the test runs (essentially, we show in \expref{Lemma}{lem:test-run-acceptance} that all non-benign Pauli attacks are detected by  one of the test runs).

\begin{lemma}
\label{lem:test-run-acceptance}
Consider the \expref{Interactive Proof System}{prot:interactive-proof-EPR} for a circuit $C$ on $n$ qubits and with $t$ $\tgate$-gate gadgets, with attack  \{$E_k$\} (with each $E_k = \sum_{Q} \alpha_{k,Q} Q$), and suppose a test run is applied.
  Let $B'_{t,n}$ be the set of non-benign attacks. Then with the following probability, the verifier rejects%
\begin{equation}
\label{eqn:acceptance-test-benign}
\frac{1}{2}\sum_k\sum_{Q \in B'_{t,n}} \abs{\alpha_{k,Q}}^2.
\end{equation}
\end{lemma}

\begin{proof}

As a first step towards proving \expref{Lemma}{lem:test-run-acceptance}, we derive an expression for the system after the application of the honest circuit (and before any attack).
 Let $h_i \in \{0,1\}$ ($i=1 \ldots t$) be a bit that indicates if the auxiliary qubit $i$ is an encrypted version of $\ket{0}$ ($h_i =0$) or $\ket{+}$ ($h_i =1$). We note that, by the properties of the protocol, $h_i=0$ in an $\xgate$-test run if and only if $h_i = 1$ in a $\zgate$-test run.

In the case of an $\xgate$-test, the system that we obtain after the verifier has performed measurements that prepare the encrypted auxiliary qubits and the computation wire is %
\begin{multline}
\label{eqn:x-test-start}
 \sum_{\substack{d_1 \ldots d_t \in \{0,1\} \\ x_1 \ldots x_t \in \{0,1\}}}
 \left(\pgate^{h_1 \cdot x_1}\hgate^{h_1} \xgate^{d_1} \egoketbra{0} \xgate^{d_1} {\hgate}^{h_1} {\pgate^*}^{h_1 \cdot x_1} \otimes
  \xgate^{x_1} \egoketbra{0} \xgate^{x_1}\right) \otimes  \cdots  \\
   \otimes  \left( \pgate^{h_t \cdot x_t} \hgate^{h_t} \xgate^{d_t}  \egoketbra{0} \xgate^{d_t} {\hgate}^{h_t} {\pgate^*}^{h_t \cdot x_t} \otimes  \xgate^{x_t} \egoketbra{0} \xgate^{x_t}\right)
\otimes  {}\\[2ex] \sum_{a_1 \ldots a_n \in \{0,1\}} \xgate^{a_1} \egoketbra{0} \xgate^{a_1} \otimes \cdots \otimes \xgate^{a_n} \egoketbra{0} \xgate^{a_n}\\
\otimes \egoketbra{d_1 \ldots d_t, x_1 \ldots x_t,a_1 \ldots a_n}\,.
\end{multline}
Note that we have appended a $2t+n$ qubit register, which is held by the verifier and that contains a classical basis state representing the key.

In the case of a $\zgate$-test, the system that we start with is %
\begin{multline}
\label{eqn:z-test-start}
 \sum_{\substack{d_1 \ldots d_t \in \{0,1\}\\ x_1 \ldots x_t \in \{0,1\}}}
\left(\pgate^{h_1 \cdot x_1}\hgate^{h_1} \xgate^{d_1} \egoketbra{0} \xgate^{d_1} {\hgate}^{h_1} {\pgate^*}^{h_1 \cdot x_1} \otimes
  \xgate^{x_1} \egoketbra{0} \xgate^{x_1}\right) \otimes  \cdots  \\
   \otimes  \left( \pgate^{h_t \cdot x_t} \hgate^{h_t} \xgate^{d_t}  \egoketbra{0} \xgate^{d_t} {\hgate}^{h_t} {\pgate^*}^{h_t \cdot x_t} \otimes  \xgate^{x_t} \egoketbra{0} \xgate^{x_t}\right)
\otimes {}\\[2ex]\sum_{b_1 \ldots b_n \in \{0,1\}} \zgate^{b_1} \egoketbra{+} \zgate^{b_1} \otimes \cdots \otimes \zgate^{b_n} \egoketbra{+} \zgate^{b_n}\\
\otimes \egoketbra{d_1 \ldots d_t, x_1 \ldots x_t,b_1 \ldots b_n}\,.
\end{multline}

We claim that, replacing the $\pgate$ and $\pgate^*$ gates with the identity in \expeqref{Equation}{eqn:x-test-start} and~\expeqref{Equation}{eqn:z-test-start}, we obtain expressions for the system for each test run, respectively, at the \emph{end} of the application of the honest unitary. This essentially follows by construction; for completeness we review the case of each gate gadget below.

 \paragraph{$\cnot$-gate.}
 As discussed in \expref{Section}{sec:CNOT-gadget}, in both the $\xgate$-test and $\zgate$-test, a $\cnot$ gate, when applied to the computation registers (the last $n$ registers), will have no effect %
 (up to a relabelling of the Paulis, %
 as computed by the key update). %
  Thus a simple change of variable reverts the system to an expression identical to its prior state.
\paragraph{$\hgate$-gate.} As given in \expref{Section}{sec:H-gadget}, the application of the $\hgate$ gate to the computation registers will, up to a relabelling of the Paulis, cause $\ket{0} \mapsto \ket{+}$ and vice-versa. Since an even number of $\hgate$ gates are applied to each computation wire, the starting input state will not be changed by these gates.

 \paragraph{$\tgate$-gate as part of an $\xgate$-test.}
\label{sec:t-gate-x-test}

Suppose a $\tgate$-gate is applied in the $\xgate$-test run, on qubit $j$, using an auxiliary qubit  $i$ $(i = 1 \ldots t)$. Suppose furthermore that qubit $j$ has undergone an \emph{even} number of $\hgate$ gates, so that %
$h_i=0$, and the system that the prover acts upon for the $\tgate$-gate gadget, together with the relevant key register is %
\begin{equation}
\label{eq:start-x-test}
\sum_{d_i,x_i,a_j \in \{0,1\}} \xgate^{d_i}\egoketbra{0}\xgate^{d_i} \otimes \xgate^{x_i}\egoketbra{0}\xgate^{x_i} \otimes \xgate^{a_j}\egoketbra{0}_j\xgate^{a_j} \otimes \egoketbra{d_i, x_i, a_j}\,.
\end{equation}
Applying the $\pgate^x$ and $\cnot$ as in the honest computation has very little effect on the system; it only causes a key update:
\begin{equation}
\sum_{d_i,x_i,a_j \in \{0,1\}} \xgate^{d_i}\egoketbra{0}\xgate^{d_i} \otimes \xgate^{x_i}\egoketbra{0}\xgate^{x_i} \otimes \xgate^{{a_j\oplus {d_i}}}\egoketbra{0}_j\xgate^{a_j \oplus d_i}\otimes \egoketbra{d_i, x_i, a_j \oplus d_i}\,.
\end{equation}

As per our convention, we swap the first and last registers, so that the data wire remains in its position:
\begin{equation}
\sum_{d_i,x_i,a_j \in \{0,1\}}\xgate^{a_j\oplus {d_i}}\egoketbra{0}\xgate^{a_j \oplus d_i}
  \otimes \xgate^{x_i}\egoketbra{0}\xgate^{x_i} \otimes\xgate^{d_i}\egoketbra{0}_j\xgate^{d_i}\otimes \egoketbra{a_j \oplus d_i,  x_i, d_i }\,.
\end{equation}

Next, a  change of variable  shows that the expression is unchanged:
\begin{equation}
\sum_{d_i,x_i,a_j \in \{0,1\}}\xgate^{d_i}\egoketbra{0}\xgate^{d_i}
  \otimes \xgate^{x_i}\egoketbra{0}\xgate^{x_i} \otimes\xgate^{a_j}\egoketbra{0}_j\xgate^{a_j}\otimes \egoketbra{d_i,  x_i, a_j }\,.
\end{equation}

\paragraph{$\tgate$-gate  as part of a $\zgate$-test.}
Suppose a $\tgate$-gate is applied in the $\zgate$ test run, on qubit $j$, using an auxiliary qubit  $i$ $(i = 1 \ldots t)$. Suppose furthermore that qubit $j$ has undergone an \emph{even} number of $\hgate$ gates, so that $h_i =1$, and the system that the prover acts upon for the $\tgate$-gate gadget,  together with the relevant key register is %
\begin{equation}
\label{eq:start-z-test}
\sum_{d_i,x_i,b_j \in \{0,1\}} \pgate^{x_i}\zgate^{d_i}\egoketbra{+}\zgate^{d_i}{\pgate^*}^{x_i}\otimes \xgate^{x_i}\egoketbra{0}\xgate^{x_i} \otimes \zgate^{b_j}\egoketbra{+}_j\zgate^{b_j}  \otimes \egoketbra{d_i, x_i, b_j}\,.
\end{equation}
Applying the $\pgate^x$ and $\cnot$ as in the honest computation  changes the system by canceling out the $\pgate^x_i$ and causing a key update:
\begin{multline}
\sum_{d_i,x_i,b_j \in \{0,1\}} \zgate^{b_j \oplus d_i \oplus x_i}\egoketbra{+}\zgate^{b_j \oplus d_i \oplus x_i}\otimes \xgate^{x_i}\egoketbra{0} \xgate^{x_i} \otimes
\zgate^{b_j} \egoketbra{+}_j\zgate^{b_j}\\  \otimes \egoketbra{b_j \oplus d_i \oplus x_i, x_i, b_j}\,.
\end{multline}

 As per our convention, we swap the first and last registers, so that the data wire remains in the last position:
 \begin{multline}
\sum_{d_i,x_i,b_j \in \{0,1\}} \zgate^{b_j}\egoketbra{+}\zgate^{b_j} \otimes \xgate^{x_i}\egoketbra{0}\xgate^{x_i} \otimes
\zgate^{b_j \oplus d_i \oplus x_i}\egoketbra{+}_j\zgate^{b_j \oplus d_i \oplus x_i}\\\otimes \egoketbra{b_j, x_i, b_j \oplus d_i \oplus x_i}\,.
\end{multline}

Next, a  change of variable shows that the expression is unchanged, except that the $\pgate$ and $\pgate^*$ gates are removed:
 \begin{equation}
\sum_{d_i,x_i,b_j \in \{0,1\}} \zgate^{d_i}\egoketbra{+}\zgate^{d_i} \otimes \xgate^{x_i}\egoketbra{0}\xgate^{x_i} \otimes
\zgate^{b_j}\egoketbra{+}_j\zgate^{b_j}\otimes \egoketbra{d_i, x_i, b_j}\,.
\end{equation}

\paragraph{Final expression before an attack.}
In the case that an \emph{odd} number of $\hgate$ gates have been applied to a data wire, the protocol specifies that we should temporarily swap the roles of the $\xgate$-test and $\zgate$-test runs for the $\tgate$-gates that   immediately follow. In this case, the data qubits and computation will be exactly those considered in the two cases above, but with the roles of the $\xgate$-test and $\zgate$-test exchanged; the same analysis thus applies. For both the $\xgate$-test and $\zgate$-test,  we iteratively apply the various cases above (depending on the circuit). Since it is the case that all computation wires eventually have an \emph{even} number of $\hgate$-gates applied, we can write down an expression for the outcome for the  $\xgate$-test run:
\begin{multline}
\label{eqn:x-test-end}
 \sum_{\substack{d_1 \ldots d_t \in \{0,1\} \\ x_1 \ldots x_t \in \{0,1\}}}
 \left(\hgate^{h_1} \xgate^{d_1} \egoketbra{0} \xgate^{d_1} {\hgate}^{h_1} \otimes
  \xgate^{x_1} \egoketbra{0} \xgate^{x_1}\right) \otimes  \cdots  \\[-2ex]
   \otimes  \left( \hgate^{h_t} \xgate^{d_t}  \egoketbra{0} \xgate^{d_t} {\hgate}^{h_t}  \otimes  \xgate^{x_t} \egoketbra{0} \xgate^{x_t}\right)
\otimes  {}\\[2ex]\sum_{a_1 \ldots a_n \in \{0,1\}} \xgate^{a_1} \egoketbra{0} \xgate^{a_1} \otimes \cdots \otimes \xgate^{a_n} \egoketbra{0} \xgate^{a_n}\\
\otimes \egoketbra{d_1 \ldots d_t, x_1 \ldots x_t,a_1 \ldots a_n}\,.
\end{multline}

In the case of a $\zgate$-test, an expression for the outcome is%
\begin{multline}
\label{eqn:z-test-end}
 \sum_{\substack{d_1 \ldots d_t \in \{0,1\}\\ x_1 \ldots x_t \in \{0,1\}}}
\left(\hgate^{h_1} \xgate^{d_1} \egoketbra{0} \xgate^{d_1} {\hgate}^{h_1} \otimes
  \xgate^{x_1} \egoketbra{0} \xgate^{x_1}\right) \otimes  \cdots  \\[-2ex]
   \otimes  \left(  \hgate^{h_t} \xgate^{d_t}  \egoketbra{0} \xgate^{d_t} {\hgate}^{h_t}  \otimes  \xgate^{x_t} \egoketbra{0} \xgate^{x_t}\right)
\otimes {}\\[2ex]\sum_{b_1 \ldots b_n \in \{0,1\}} \zgate^{b_1} \egoketbra{+} \zgate^{b_1} \otimes \cdots \otimes \zgate^{b_n} \egoketbra{+} \zgate^{b_n}\\
\otimes \egoketbra{d_1 \ldots d_t, x_1 \ldots x_t,b_1 \ldots b_n}\,.
\end{multline}

\paragraph{Applying the attack, decryption and measurement.} Next, we apply the attack for a fixed~$k$,  as given by
\begin{equation}
  E_k = \sum_Q \alpha_{Q} Q\qquad\text{and}\qquad E_k^* = \sum_{Q'} \alpha_{Q'}^* Q'\,,
\end{equation}
followed by the verifier's decryption, trace and measurement.  For the
registers that are traced out, we assume that they are decrypted and
measured. %
 Furthermore, since they are traced out, we can assume that
the quantum auxiliary registers with $h_i =1$ are measured in the
diagonal basis.  We let
\begin{equation}
  Q=P_1\otimes Q_1 \otimes P_2 \otimes Q_2 \otimes \cdots \otimes \P_t
  \otimes Q_t \otimes R_1 \otimes \cdots \otimes R_n\,,
\end{equation}
with $P_i, Q_i, R_j \in \mathbb{P}_1$ ($ i =1 \ldots t,j=1 \ldots n$)
and similarly, let
\begin{equation}
Q'={P'}_1\otimes {Q'}_1 \otimes {P'}_2 \otimes {Q'}_2 \otimes \cdots
\otimes {P'}_t \otimes {Q'}_t \otimes {R'}_1 \ldots {R'}_n\,,
\end{equation}
with ${P'}_i, {Q'}_i, {R'}_j \in \mathbb{P}_1$
($ i =1 \ldots t,j=1 \ldots n$).

 For the $\xgate$-test run, conditioned on outcomes $i_\ell$ (where $h_\ell =0$), and
   $k_1, \ldots ,k_n$ the system becomes%
\begin{multline}
\sum_{\substack{ (h_\ell=1 \wedge \\  i_\ell \in \{0,1\}), \\
j_m \in \{0,1\}}}\sum_{Q, Q'} \alpha_Q \alpha^*_{Q'} \sum_{\substack{d_1 \ldots d_t \in \{0,1\} \\ x_1 \ldots x_t \in \{0,1\}}}\\
 \left( \bra{i_1} \xgate^{d_1} \hgate^{h_1} P_1 \hgate^{h_1} \xgate^{d_1} \egoketbra{0} \xgate^{d_1} {\hgate}^{h_1} {P'}_1 {\hgate}^{h_1} \xgate^{d_1} \ket{i_1} \otimes
\bra{j_1} \xgate^{x_1} Q_1 \xgate^{x_1} \egoketbra{0} \xgate^{x_1} {Q'}_1 \xgate^{x_1} \ket{j_1} \right)\\
 \otimes  \cdots
   \otimes
     \\\left(  \bra{i_t} \xgate^{d_t} \hgate^{h_t}  P_t \hgate^{h_t}   \xgate^{d_t}  \egoketbra{0} \xgate^{d_t} {\hgate}^{h_t}  {P'}_t {\hgate}^{h_t} \xgate^{d_t} \ket{i_t}  \otimes \bra{j_t} \xgate^{x_t}  Q_t \xgate^{x_t}  \egoketbra{0} \xgate^{x_t} {Q'}_t \xgate^{x_t} \ket{j_t} \right)
\otimes  {}\\[2ex]\sum_{a_1 \ldots a_n \in \{0,1\}}\bra{k_1} \xgate^{a_1}R_1 \xgate^{a_1} \egoketbra{0} \xgate^{a_1}{R'}_1 \xgate^{a_1} \ket{k_1} \otimes \cdots \otimes \bra{k_n} \xgate^{a_n}  R_n \xgate^{a_n} \egoketbra{0} \xgate^{a_n} R'_n \xgate^{a_n} \ket{k_n}\,.
\end{multline}

Applying the classical Pauli twirl (\expref{Lemmas}{lem:classical-Pauli-twirl-X} and~\ref{lem:classical-Pauli-twirl-Z}), we obtain that the cross terms of the attack vanish, leaving as expression%
\begin{multline}
\sum_{\substack{(h_\ell=1 \wedge \\  i_\ell \in \{0,1\}), \\
j_m \in \{0,1\}}} \sum_{Q} \abs{\alpha_Q}^2
 \left( \bra{i_1} \hgate^{h_1} P_1 \hgate^{h_1}  \egoketbra{0}  {\hgate}^{h_1} {P}_1 {\hgate}^{h_1}  \ket{i_1} \otimes
\bra{j_1}  Q_1  \egoketbra{0}  {Q}_1  \ket{j_1} \right) \otimes  \cdots  \\[-3ex]
   \otimes  \left(  \bra{i_t}  \hgate^{h_t}  P_t \hgate^{h_t}     \egoketbra{0} \xgate^{d_t} {\hgate}^{h_t}  {P}_t {\hgate}^{h_t}  \ket{i_t}  \otimes \bra{j_t}   Q_t   \egoketbra{0}  {Q}_t  \ket{j_t} \right)
\otimes  {}\\[2ex]\bra{k_1} R_1  \egoketbra{0} {R}_1  \ket{k_1} \otimes \cdots \otimes \bra{k_n}   R_n \egoketbra{0}  R_n  \ket{k_n}\,.
\end{multline}

Recall that in an $\xgate$-test run, the verifier rejects if a measurement result on an auxiliary qubit with $h_i=0$ decrypts to the value~$1$, or if the output decrypts to the value~$1$. Thus, applying the above to all terms in $\{E_k\}$, we get that  the probability that the verifier rejects is given by
\begin{equation}
  \sum_k\sum_{Q \in B_1}\abs{\alpha_{k,Q}}^2\,,
\end{equation}
where $B_1$ is the set of $2t+n$-qubit Paulis with
$P_i \in \{\xgate, \ygate\}$ $(i = 1\ldots t)$ whenever $h_i =0$, or
with $R_n \in \{\xgate, \ygate\}$.

A similar calculation shows that for the $\zgate$-test run,
conditioned on outcomes $i_\ell$ (where $h_\ell =1$), and
   $k_1, \ldots ,k_n$ the system becomes%
\begin{multline}
\sum_{\substack{(h_\ell=0 \wedge \\  i_\ell \in \{0,1\}), \\
j_m \in \{0,1\}}} \sum_{Q} \abs{\alpha_Q}^2
 \left( \bra{i_1} \hgate^{h_1} P_1 \hgate^{h_1}  \egoketbra{0}  {\hgate}^{h_1} {P}_1 {\hgate}^{h_1}  \ket{i_1} \otimes
\bra{j_1}  Q_1  \egoketbra{0}  {Q}_1  \ket{j_1} \right) \otimes  \cdots  \\[-3ex]
   \otimes  \left(  \bra{i_t}  \hgate^{h_t}  P_t \hgate^{h_t}     \egoketbra{0} \xgate^{d_t} {\hgate}^{h_t}  {P}_t {\hgate}^{h_t}  \ket{i_t}  \otimes \bra{j_t}   Q_t   \egoketbra{0}  {Q}_t  \ket{j_t} \right)
\otimes  {}\\[2ex]\bra{k_1} R_1  \egoketbra{+} {R}_1  \ket{k_1} \otimes \cdots \otimes \bra{k_n}   R_n \egoketbra{+}  R_n  \ket{k_n}\,.
\end{multline}

Recall that in a $\zgate$-test run, the verifier rejects if a measurement result on an auxiliary qubit with $h_i=0$ decrypts to the value~$1$.
Thus, applying the above to all terms in $\{E_k\}$, we get that the probability that the verifier rejects is given by
\begin{equation}
  \sum_k\sum_{Q \in B_2}\abs{\alpha_{k,Q}}^2\,,
\end{equation}
where $B_2$ is the set of $2t+n$-qubit Paulis with
$P_i \in \{\xgate, \ygate\}$ $(i = 1\ldots t)$ whenever $h_i =0$.

Since each test is executed  with probability ${1}/{2}$, and since  $h_i = 0$ in the $\xgate$-test run if and only if $h_i =1$ in the $\zgate$-test run, we obtain that the
probability that the verifier rejects is%
\begin{equation}
\frac{1}{2} \sum_k\sum_{Q \in B'_{t,n}} \abs{\alpha_{k,Q}}^2\,. \qedhere
\end{equation}
 \end{proof}

\subsection{In the case of a computation run}
\label{sec:soundness-computation run}

 Again using the preliminaries of \expref{Section}{sec:simplifying-general} and \expref{Section}{sec:conventions}, we now analyze soundness in the case of a computation run. First, we determine the effect of a bit flip on the measured qubit in the $\tgate$-gate gadget (\expref{Section}{sec:effect-X-gate-output}), then we do an analysis for the case that the computation run consists in
a single $\tgate$-gate gadget (\expref{Section}{sec:soundness-T-gate-Paulis}).
  This is extended to full generality in \expref{Section}{sec:soundness-computation-general}, where we give a lower bound (as a function of the attack and of the underlying computation) on the probability that the verifier rejects in a computation run (see \expref{Lemma}{lem:acceptance-computation}).

\subsubsection{Effect of a bit flip on a measured qubit}
\label{sec:effect-X-gate-output}
In \expref{Lemma}{lem:error-t-gate}, we establish the effect of a bit flip on the measured qubit in the $\tgate$-gate gadget.

\begin{lemma} \label{lem:error-t-gate}
The error induced by an $\xgate$-gate on the measured qubit in the $\tgate$-gate gadget in \expref{Figure}{fig:R-gate-EPR-computation} is to  introduce an extra $\xgate \zgate \pgate$ on the output.
\end{lemma}
\begin{proof}
An $\xgate$-gate on the measured qubit in \expref{Figure}{fig:R-gate-computation} will cause the bottom wire to receive the correction $\pgate^{a\oplus c \oplus y \oplus 1}$ (instead of $\pgate^{a \oplus c \oplus y}$). Since $\pgate^{a \oplus c \oplus y \oplus 1} =\pgate \zgate^{a \oplus c \oplus y} \pgate^{a \oplus c \oplus y }$,
the following shows how we can use 
 we use and revise the calculation from \expeqref{Equation}{eqn:circuit-3}--\expeqref{Equation}{eqn:circuit-35}.%
\begin{align} \label{Eqn:bit-flipstart}
 \pgate^{a \oplus c \oplus y \oplus 1 } \xgate^d  \zgate^e \pgate^y \tgate \xgate^{a\oplus c}\zgate^b
&= \pgate \zgate^{a \oplus c \oplus y} \pgate^{a \oplus c \oplus y } \xgate^d  \zgate^e \pgate^y \tgate \xgate^{a\oplus c}\zgate^b \\
& = \pgate \zgate^{a \oplus c \oplus y} \xgate^{a\oplus c \oplus d}  \zgate^{(a \oplus c \oplus d)\cdot (d \oplus y)   \oplus a \oplus b \oplus c \oplus d \oplus e  \oplus y   } \tgate \\
&= \xgate^{a\oplus c \oplus d}  \zgate^{(a \oplus c \oplus d)\cdot (d \oplus y)   \oplus a \oplus b \oplus c \oplus d \oplus e  \oplus y   } \zgate^{a \oplus c \oplus y} \zgate^{a\oplus c \oplus d}   \pgate \tgate\\ \label{eq:effect-bit-flip-last}
&= \xgate^{a\oplus c \oplus d}  \zgate^{(a \oplus c \oplus d)\cdot (d \oplus y)   \oplus a \oplus b \oplus c \oplus e }  \pgate \tgate\,. %
\end{align}
We note furthermore that the $\xgate$-gate on the measured qubit causes the Pauli key to be updated as $c\leftarrow c \oplus 1$. Starting with the right-hand side of
\expeqref{Equation} {eq:effect-bit-flip-last}, we thus obtain%
\begin{equation}\label{Eqn:bit-flipend}
\xgate^{a\oplus c \oplus d \oplus 1}  \zgate^{(a \oplus c \oplus d \oplus 1)\cdot (d \oplus y)   \oplus a \oplus b \oplus c \oplus e \oplus 1 }  \pgate \tgate
= \xgate^{a\oplus c \oplus d }  \zgate^{(a \oplus c \oplus d )\cdot (d \oplus y)   \oplus a \oplus b \oplus c \oplus d \oplus e \oplus y} (\xgate \zgate \pgate) \tgate\,.
\end{equation}
Comparing with \expeqref{Equation} {eqn:circuit-35}, we thus see that the effect is to apply  $\xgate \zgate \pgate$.
\end{proof}

\subsubsection{The \texorpdfstring{$\tgate$}{T}-gate protocol under attack}
\label{sec:soundness-T-gate-Paulis}

In this section, we  analyze the effect of an attack in a single
$\tgate$-gate gadget: we show that the effect of the gadget on an encrypted qubit is to apply a $\tgate$ gate on the plaintext, while maintaining the encryption. Furthermore, if the auxiliary qubit undergoes an attack $Q_1 \in \{\xgate, \ygate\}$, then an error $E=\xgate\zgate\pgate$ (by \expref{Lemma}{lem:error-t-gate}) will be applied to the computation wire. The Pauli twirl plays again an important part in simplifying a general attack to a convex combination of Pauli attacks.
This analysis (\expref{Lemma}{lem:single-t-gate}) is used as the inductive step in the proof of the general case as analysed in \expref{Section}{sec:soundness-computation-general} (see \expref{Lemma}{lem:acceptance-computation-induction}).

\begin{lemma}
\label{lem:single-t-gate}
In \expref{Interactive Proof System}{prot:interactive-proof-EPR}, consider a circuit consisting in a
single $\tgate$-gate, applied to a data wire which is initially an encryption of $\rho$. Consider a term in the attack \{$E_k$\}, given by $Q=P_1\otimes Q_1$, $Q' = {P'}_1 \otimes {Q'}_1$, acting on the auxiliary qubits ($P_1, Q_1, {P'}_1, {Q'}_1 \in \mathbb{P}_1$).  Then in the case $P_1\otimes Q_1 \neq P'_1 \otimes Q'_1$, this term simplifies to~$0$, whereas otherwise the effect is to apply $E^{\delta_{P_1}} \tgate$ on the data, while maintaining a uniform encryption, with the key held by the verifier. (Here,  $Q \in \mathbb{P}_1$, and $\delta_{Q} = 0$ if $Q \in \{I, \zgate\} $ and $\delta_{Q} = 1$ otherwise.)
\end{lemma}

\begin{proof}
 Suppose the honest circuit of the prover is applied, followed by an attack and the coherent correction of the verifier, which is then followed by a measurement of auxiliary qubits. We consider the effect of a Pauli attack $Q=P_1\otimes Q_1$, $Q' = {P'}_1 \otimes {Q'}_1$, acting on the auxiliary qubits ($P_1, Q_1, {P'}_1, {Q'}_1 \in \mathbb{P}_1$). 
  (Strictly 
  speaking this is not a full attack, but instead, ignoring the coefficient, it corresponds to one term in the expansion of the full attack as given by $\{E_k\}$%
  .)
When a bit flip occurs on the top wire (\ie,~when $P_1 \in \{\xgate, \ygate\}$) in \expref{Figure}{fig:R-gate-EPR-computation}, then the outcome undergoes an error $E=\xgate\zgate\pgate$ as given by \expref{Lemma}{lem:error-t-gate}. The above is summarized in  \expref{Figure}{fig:R-gate-EPR-computation-attack}.

\begin{figure}[H]
\centerline{
 \Qcircuit @C=1em @R=1em  {
&&&&&&&\lstick{\xgate^a \zgate^b\ket{\psi}} & \qw & \qw & \qw & \targ & \gate{P_1}  &  \meter & \control \cw \cwx[4]  &    &   \\
\lstick{\text{Prover}}&&&&& & && & \control  & \cw    & \cw
&   \cgate{Q_1} & \control \cw  & &&\\
&&&&&& & & && & \ctrl{-2} &\qw    & \gate{\pgate^x} \cwx &  \qw  &
\qw& \rstick{ \hspace{-0.5cm}\xgate^{a'}  \zgate^{b'} E^{\delta_{P_1}} \tgate   \ket{\psi}}\\
&&&&\dw &\dw&\dw& \dw & \dw & \dw &\dw & \dw
 &  \dw & \dw  & \dw &\dw&\\
&& &&&& &\lstick{x \in_R \{0,1\}} &\cw & \control \cw \cwx[-3] & &
&   & & \control &   \cw  & \rstick{c} \\
\lstick{\text{Verifier}}&&& &&&&\lstick{\ket{0}}    & \gate{\hgate}
&\ctrl{1} & \qw \qwx[-3]
 & &  &\mbox{($ d \in_R \{0,1\}, y = a \oplus c \oplus x $)}  & & &&\\
&& &&&& &\lstick{\ket{0}}    & \qw & \targ & \gate{\rgate}
 & \gate{\pgate^{y \oplus d}}  & \gate{\zgate^{d}}& \gate{\hgate} &  \meter &  \cw  & \cw
   & \rstick{e}
    \gategroup{6}{6}{7}{10}{1.4em}{--}
     \gategroup{1}{2}{3}{2}{0.7em}{\{}
 \gategroup{5}{2}{7}{2}{0.7em}{\{}
 }
 }
  \caption{\label{fig:R-gate-EPR-computation-attack} An attack $P_1 \otimes Q_1$ on a single $\tgate$-gate gadget (computation run). Here, $a'= a\oplus c$ and $b'=(a \oplus c) \cdot (d \oplus y) \oplus a \oplus  b \oplus c \oplus e \oplus y $\,.%
 }
\end{figure}
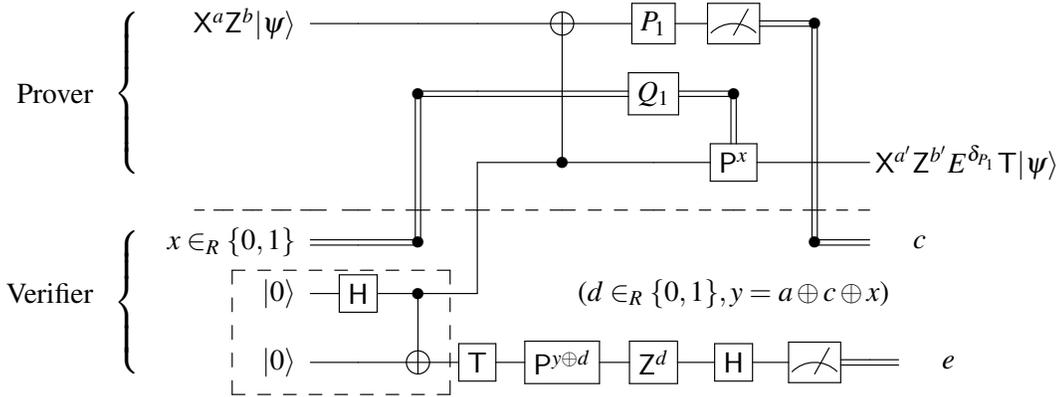

In order to give a mathematical expression for \expref{Figure}{fig:R-gate-EPR-computation-attack}, we use the circuit identity in \expref{Figure}{fig:circuit-proof:5} (according to which the measurement outcome $c$ undergoes an encryption and decryption with~$d$).  First we consider the case
$\delta_{P_1} = \delta_{P'_1}$. Applying  $ y =a \oplus c \oplus x$, and considering the trace of the auxiliary qubits, we get the following
 expression %
\begin{multline} \label{eqn:t-gate-first}
 \sum_{i,j \in \{0,1\}}\sum_{a,b,c,d,e,x \in \{0,1\}}  \bra{i}  \xgate^d P_1 \xgate^d \ketbra{c}{c} \xgate^d {{P'}_1} \xgate^d  \ket{i}  \otimes
  \bra{j} \xgate^x Q_1  \xgate^x \egoketbra{0} \xgate^x {{Q'}_1} \xgate^x \ket{j} \otimes
 \\\xgate^{a\oplus c}  \zgate^{(a \oplus c) \cdot (d \oplus x) \oplus a \oplus  b \oplus c \oplus e \oplus x}  E^{\delta_{P_1}} \tgate \rho \tgate^*  {E^*}^{\delta_{P_1}}
    \zgate^{(a \oplus c) \cdot (d \oplus x) \oplus a \oplus  b \oplus c \oplus e \oplus x}  \xgate^{a\oplus c} \\[1ex]
 \otimes \ket{a \oplus c, (a \oplus c) \cdot (d \oplus x) \oplus a \oplus b \oplus c \oplus e \oplus x  } \bra{a \oplus c,(a \oplus c) \cdot (d \oplus x) \oplus a \oplus b \oplus c \oplus e \oplus x}\,.
 \end{multline}
 Note that above, we have included a key register for the computation wire, and have considered that the first register is traced out, thus since $\delta_{P_1} = \delta_{P'_1}$, cross terms of the form
 \begin{equation}
   \bra{i}  \xgate^d P_1 \xgate^d \ketbra{c}{c'} \xgate^d {{P'}_1} \xgate^d  \ket{i}
 \end{equation}
 with $c \neq c'$ vanish and are therefore excluded.  Furthermore, we
 consider without loss of generality that the second register is
 decrypted with $\xgate^x$ before being traced out.

 Since $a$, $b$, and $e$ appear only on the data register, and by a change of variable,  we can
rewrite  %
this as
\begin{multline} \label{eqn:t-gate-first-ab}
 \sum_{i,j \in \{0,1\}} \sum_{c,d,x \in \{0,1\}}  \bra{i}  \xgate^d P_1 \xgate^d \ketbra{c}{c} \xgate^d {{P'}_1} \xgate^d  \ket{i}  \otimes
  \bra{j} \xgate^x  Q_1
 \xgate^x \egoketbra{0} \xgate^x {{Q'}_1}\xgate^x \ket{j}  \otimes
 \\ \sum_{a,b\in \{0,1\}}  \xgate^{a}  \zgate^{b}  E^{\delta_{P_1}} \tgate \rho \tgate^*  {E^*}^{\delta_{P_1}}
    \zgate^{b}  \xgate^{a}
 \otimes \ket{a,b} \bra{a,b}\,.
 \end{multline}
Applying the classical Pauli twirl (\expref{Lemma}{lem:classical-Pauli-twirl-X}), we get that for the case  under consideration ($\delta_{P_1} = \delta_{P'_1}$), the system is $0$ if $P_1 \otimes Q_1 \neq P'_1 \otimes Q'_1$, and otherwise is%
\begin{equation} \label{eqn:simplified-t-gate}
\sum_{a,b\in \{0,1\}}  \xgate^{a}  \zgate^{b}  E^{\delta_{P_1}} \tgate \rho \tgate^*  {E^*}^{\delta_{P_1}}
    \zgate^{b}  \xgate^{a}
 \otimes \ket{a,b} \bra{a,b}\,.
\end{equation}

 Next we consider the case $\delta_{P_1} \neq \delta_{P'_1}$.
  Again applying  $ y =a \oplus c \oplus x$, letting $c' = c \oplus 1$, and considering the trace of the auxiliary qubits, we get the following
 expression%
\begin{multline} \label{eqn:t-gate-first-ab2}
\sum_{i,j \in \{0,1\}}\sum_{a,b,c,d,e,x \in \{0,1\}}   \bra{i} \xgate^d P_1 \xgate^d \ketbra{c}{c'} \xgate^d {{P'}_1} \xgate^d \ket{i}   \otimes
\bra{j} \xgate^x Q_1
 \xgate^x \egoketbra{0} \xgate^x {{Q'}_1} \xgate^x \ket{j} \otimes
 \\ \xi  \xgate^{a\oplus c}  \zgate^{(a \oplus c) \cdot (d \oplus x) \oplus a \oplus  b \oplus c \oplus e \oplus x}  E^{\delta_{P_1}} \tgate \rho \tgate^*  {E^*}^{\delta_{P'_1} }
    \zgate^{(a \oplus c') \cdot (d \oplus x) \oplus a \oplus  b \oplus c' \oplus e \oplus x}  \xgate^{a\oplus c'} \\[1ex]
 \otimes \ket{a \oplus c, (a \oplus c) \cdot (d \oplus x) \oplus a \oplus b \oplus c \oplus e \oplus x  } \bra{a \oplus c' ,(a \oplus c') \cdot (d \oplus x) \oplus a \oplus b \oplus c' \oplus e \oplus x}\,.
 \end{multline}
Note that above, we have included a key register for the computation wire, and have considered that the first register is traced out. %
 Thus terms of the form
$\bra{i}  \xgate^d P_1 \xgate^d \ketbra{c}{c} \xgate^d {{P'}_1} \xgate^d  \ket{i}$ vanish and are therefore excluded.
We again consider without loss of generality that the second register is decrypted with $\xgate^x$ before being traced out. Here, $\xi$ represents a phase that depends on the key register and is due to the fact that $c \neq c'$.

Since $\delta_{P_1} =  \delta_{P'_1} \oplus 1$, since  $a$, $b$, and $e$ appear only on the data register, and by the following change of variable,
\begin{align}
 a &\leftarrow a \oplus c\\
 b & \leftarrow  (a \oplus c) \cdot (d \oplus x) \oplus a \oplus  b \oplus c \oplus e \oplus x \\
 e & \leftarrow  (a \oplus c') \cdot (d \oplus x) \oplus a \oplus b \oplus c' \oplus e \oplus x\,,
\end{align}
we can
rewrite  %
this as
\begin{multline} \label{eqn:t-gate-first-ab3}
\sum_{i,j \in \{0,1\}}\sum_{c,d,x \in \{0,1\}}   \bra{i} \xgate^d P_1 \xgate^d \ketbra{c}{c} \xgate^d (\xgate {{P'}_1}) \xgate^d \ket{i}   \otimes
\bra{j} \xgate^x Q_1  \xgate^x \egoketbra{0} \xgate^x {{Q'}_1} \xgate^x \ket{j} \otimes
 \\  \sum_{a,b,e \in \{0,1\}} \xi  \xgate^{a}\zgate^{b}  E^{\delta_{P_1}} \tgate \rho \tgate^*  {E^*}^{\delta_{P_1}}
    \zgate^{e}  \xgate^{a \oplus 1}
 \otimes \ket{a,b} \bra{a \oplus 1,e }\,.
 \end{multline}

Next, 
we
 apply 
 the 
 Pauli twirl to the first register, which yields $0$ unless $P_1 = \xgate P'_1$, in which case we also get~$0$,
since $ \bra{i} P_1  \ketbra{c}{c}  (\xgate {P_1})  \ket{i}=0 $.

Thus, we conclude that the case $P_1\otimes Q_1 \neq P'_1 \otimes Q'_1$ evaluates to~$0$, whereas otherwise (by \expeqref{Equation} {eqn:simplified-t-gate}) the effect of the $\tgate$-gate gadget is to apply $E^{\delta_{P_1}} \tgate$ on the data, while maintaining a uniform encryption, with the key held by the verifier.
\end{proof}

\subsubsection{General analysis for a computation run}
\label{sec:soundness-computation-general}

In this section, we  bound the acceptance probability in the case of a computation run. %
  This result is presented in \expref{Lemma}{lem:acceptance-computation}, the proof of which depends on the following Lemma%

\begin{lemma}\label{lem:acceptance-computation-induction}
Consider \expref{Interactive Proof System}{prot:interactive-proof-EPR} for a circuit $C$ on $n$ qubits and with $t$ $\tgate$-gate gadgets, with attack $\{E_k\}$ (with each $E_k = \sum_{Q} \alpha_{k,Q} Q$).
Suppose the target circuit  $C$ is decomposed as
\begin{equation}
  C =  \tgate_{\ell_t} C_t \ldots \tgate_{\ell_2} C_2 \tgate_{\ell_1} C_1\,,
\end{equation}
where each $C_i$ is a Clifford group circuit and
$\ell_i \in \{1\ldots n\}$ indicates that the $i^{\text{th}}$
$\tgate$-gate acts on qubit $\ell_i$ $ (i=1\ldots t)$.
Let
\begin{equation}
  Q=P_1\otimes Q_1 \otimes P_2 \otimes Q_2 \otimes \cdots \otimes \P_t \otimes Q_t\,,
\end{equation}
with $P_i, Q_i \in \mathbb{P}_1$ ($ i =1 \ldots t$) and similarly, let
\begin{equation}
  Q'={P'}_1\otimes {Q'}_1 \otimes {P'}_2 \otimes {Q'}_2 \otimes \cdots
  \otimes {P'}_t \otimes {Q'}_t\,,
\end{equation}
with ${P'}_i, {Q'}_i \in \mathbb{P}_1$ ($ i =1 \ldots t$).  Let $E$ be
the error on the output induced by an $\xgate$ on the top wire in
\expref{Figure}{fig:R-gate-EPR-computation}.  For a
Pauli~$P \in \mathbb{P}_1$, let $\delta_{P} = 0$ if
$P \in \{I, \zgate\} $ and $\delta_{P} = 1$ otherwise.  Then we claim
that after the honest computation, attack, the verifier's conditional
corrections and tracing out of the auxiliary system, the term of the
system corresponding to the attack $(Q, Q')$ is~$0$ if
\begin{equation}
  P_1\otimes Q_1 \otimes P_2 \otimes Q_2 \otimes \cdots \otimes \P_t
\otimes Q_t \neq P'_1\otimes Q'_1 \otimes P'_2 \otimes Q'_2 \otimes
\cdots \otimes \P'_t \otimes Q'_t\,,
\end{equation}
and otherwise is%
\begin{multline}
\label{eqn:computation-run-claim}
\sum_{\substack{a_1,b_1, \ldots a_n, b_n \\ \in \{0,1\}}} (\xgate^{a_1}\zgate^{b_1} \otimes \cdots \otimes \xgate^{a_n}\zgate^{b_n})  E^{\delta_{P_t}}_{\ell_t}    \tgate_{\ell_t} C_t \ldots  E^{\delta_{P_2}}{\ell_2}    \tgate_{\ell_2} C_2  E^{\delta_{P_1}}_{\ell_1}  \tgate_{\ell_1} C_1 \ket{0^n}\\
 \bra{0^n} {C_1}^* \tgate^*_{\ell_1}   {E^*}^{\delta_{P_1}}_{\ell_1} {C_2}^*   \tgate^*_{\ell_2}   {E^*}^{\delta_{P_2}}_{\ell_2} \ldots   {C_t}^* \tgate^*_{\ell_t}   {E^*}^{\delta_{P_t}}_{\ell_t}  (\zgate^{b_1} \xgate^{a_1}\otimes \cdots \zgate^{b_n} \xgate^{a_n})\\[2ex]
 \otimes \egoketbra{a_1, b_1, \ldots ,a_n, b_n}\,.
\end{multline}
\end{lemma}

Note that in the statement of \expref{Lemma}{lem:acceptance-computation-induction} above, we have that $Q$ and $Q'$ are applied only to auxiliary qubits---we have thus omitted the ``$R$'' portion of the attack. %
 Furthermore, we refer the reader to \expref{Figure}{fig:R-gate-EPR-computation-attack} for the $t=1$ case.

\begin{proof}
\expref{Lemma}{lem:acceptance-computation-induction} is proven by induction on $t$. The base case $t=0$ is verified, since the initial system~is%
\begin{equation}
\sum_{\substack{a_1,b_1, \ldots ,a_n, b_n \\ \in \{0,1\}}}  \xgate^{a_1}\zgate^{b_1} \otimes \cdots \otimes \xgate^{a_n}\zgate^{b_n} \egoketbra{0^n}  \zgate^{b_1} \xgate^{a_1}\otimes \cdots \otimes \zgate^{b_n} \xgate^{a_n} \otimes \egoketbra{a_1, b_1, \ldots ,a_n, b_n}\,.
\end{equation}

Next, suppose \expeqref{Equation} {eqn:computation-run-claim} holds for $t=k$ and consider $t=k+1$.
Now, because each Pauli acts on a different subsystem,  an attack
\begin{equation}
  Q=P_1\otimes Q_1 \otimes P_2 \otimes Q_2 \otimes \cdots \otimes \P_k \otimes Q_k \otimes \P_{k+1} \otimes Q_{k+1}
\end{equation}
can be decomposed as an attack on the first $2k$ auxiliary qubits,
followed by an attack of the last two qubits (and similarly for
$Q'$). Suppose $P_i \otimes Q_i = P'_i \otimes Q'_i$ $(i=1\ldots
k$). By the hypothesis, the outcome will be the application of the
computation corresponding to the gadgets for $C_{k+1}$ and
$T_{\ell_{k+1}}$, followed by attack
\begin{equation}
  (Q_{k+1} \otimes P_{k+1}, Q'_{k+1} \otimes P'_{k+1})\,,
\end{equation}
on an encrypted system~$\rho$ given by:
\begin{multline}
\rho = \sum_{\substack{a_1,b_1, \ldots a_n, b_n \\ \in \{0,1\}}} (\xgate^{a_1}\zgate^{b_1} \otimes \cdots \otimes \xgate^{a_n}\zgate^{b_n})  E^{\delta_{P_k}}_{\ell_k}    \tgate_{\ell_k} C_k \ldots  E^{\delta_{P_2}}{\ell_2}    \tgate_{\ell_2} C_2  E^{\delta_{P_1}}_{\ell_1}  \tgate_{\ell_1} C_1 \ket{0^n}\\
 \bra{0^n} C_1^* \tgate^*_{\ell_1}   {E^*}^{\delta_{P_1}}_{\ell_1} C_2^*   \tgate^*_{\ell_2}   {E^*}^{\delta_{P_2}}_{\ell_2} \ldots   C_k^* \tgate^*_{\ell_k}   {E^*}^{\delta_{P_t}}_{\ell_k}  (\zgate^{b_1} \xgate^{a_1}\otimes \cdots \zgate^{b_n} \xgate^{a_n})
 \otimes {}\\[2ex] \egoketbra{a_1, b_1, \ldots ,a_n, b_n}\,.
\end{multline}

First, the prover applies a Clifford circuit $C_{k+1}$. After a key update as given in \expref{Section}{sec:Interactive-Proof-aux}, a change of variable will lead to~$\rho'$:
\begin{multline}
\rho' = \sum_{\substack{a_1,b_1, \ldots a_n, b_n \\ \in \{0,1\}}} (\xgate^{a_1}\zgate^{b_1} \otimes \cdots \otimes \xgate^{a_n}\zgate^{b_n}) {C_{k+1}} E^{\delta_{P_k}}_{\ell_k}    \tgate_{\ell_k} C_k \ldots  E^{\delta_{P_2}}{\ell_2}    \tgate_{\ell_2} C_2  E^{\delta_{P_1}}_{\ell_1}  \tgate_{\ell_1} C_1 \ket{0^n}\\[-1ex]
 \bra{0^n} {C_1}^* \tgate^*_{\ell_1}   {E^*}^{\delta_{P_1}}_{\ell_1} C_2^*   \tgate^*_{\ell_2}   {E^*}^{\delta_{P_2}}_{\ell_2} \ldots   C_k^* \tgate^*_{\ell_k}   {E^*}^{\delta_{P_t}}_{\ell_k} C_{k+1}^* (\zgate^{b_1} \xgate^{a_1}\otimes \cdots \zgate^{b_n} \xgate^{a_n})
 \otimes {}\\[2ex]
 \egoketbra{a_1, b_1, \ldots ,a_n, b_n}\,.
\end{multline}

Next, the auxiliary qubits for  $\tgate_{\ell_{k+1}}$ are used; we apply \expref{Lemma}{lem:single-t-gate} (only a single qubit, $\ell_{k+1}$ is affected by this part of the computation). Thus, if
\begin{equation}
  Q_{k+1} \otimes P_{k+1} \neq Q'_{k+1} \otimes P'_{k+1}\,,
\end{equation}
the term vanishes, and otherwise it becomes%
\begin{multline}
\sum_{\substack{a_1,b_1, \ldots a_n, b_n \\ \in \{0,1\}}} (\xgate^{a_1}\zgate^{b_1} \otimes \cdots \otimes \xgate^{a_n}\zgate^{b_n})  E^{\delta_{P_{k+1}}}_{\ell_{k+1}}    \tgate_{\ell_{k+1}} {C_{k+1}} E^{\delta_{P_k}}_{\ell_k}    \tgate_{\ell_k} C_k \ldots  E^{\delta_{P_2}}{\ell_2}    \tgate_{\ell_2} C_2  E^{\delta_{P_1}}_{\ell_1}  \tgate_{\ell_1} C_1 \ket{0^n}\\
 \bra{0^n} C_1^* \tgate^*_{\ell_1}   {E^*}^{\delta_{P_1}}_{\ell_1} C_2^*   \tgate^*_{\ell_2}   {E^*}^{\delta_{P_2}}_{\ell_2} \ldots   C_k^* \tgate^*_{\ell_k}   {E^*}^{\delta_{P_t}}_{\ell_k} C_{k+1}^*  {\tgate^*}_{\ell_{k+1}}  {E^*}^{\delta_{P_{k+1}}}_{\ell_{k+1}}
  (\zgate^{b_1} \xgate^{a_1}\otimes \cdots \zgate^{b_n} \xgate^{a_n})
 \otimes {}\\[2ex]
 \egoketbra{a_1, b_1, \ldots ,a_n, b_n}\,.
\end{multline}
By the hypothesis, the term also vanishes if
\begin{equation}
  P_i\otimes Q_i \neq P'_i \otimes Q'_i\qquad (i=1\ldots k)\,.
\end{equation}
Thus, \expref{Lemma}{lem:acceptance-computation-induction} holds for
the case $t= k+1$ and by induction,
\expref{Lemma}{lem:acceptance-computation-induction} holds in general.
\end{proof}

We now state and prove the main result of this Section.

\begin{lemma}\label{lem:acceptance-computation}
Consider the \expref{Interactive Proof System}{prot:interactive-proof-EPR} for a circuit $C$ on $n$ qubits and with $t$ $\tgate$-gate gadgets, with attack $\{E_k\}$ (with each $E_k = \sum_{Q} \alpha_{k,Q} Q$).
   Let $B_{t,n}$ be the set of benign attacks. Then the probability  that the verifier rejects for a computation run is at least:
\begin{equation}
\label{eqn:acceptance-computation-benign}
(1-p) \sum_k \sum_{Q \in B_{t,n}} \abs{\alpha_{k,Q}}^2\,,
\end{equation}
where
\begin{equation}
  p= \norm{(\ket{0}\bra{0}\otimes \mathbb{I}_{n-1}) C \ket{0^n}}^2
\end{equation}
is the probability that we observe ``0'' as a results of a
computational basis measurement of the $n^{\text{th}}$ output qubit,
obtained by evaluating~$C$ on input~$\ket{0^n}$.
\end{lemma}

\begin{proof}
Let the notation be as in \expref{Lemma}{lem:acceptance-computation-induction}.
We apply \expref{Lemma}{lem:acceptance-computation-induction} to the case of the attack $\{E_k\}$. We denote each
\begin{equation}
  E_k = \sum_{Q\otimes R} \alpha_{k,Q \otimes R} Q\otimes R\,,
\end{equation}
where
\begin{equation}
  Q=P_1\otimes Q_1 \otimes P_2 \otimes Q_2 \otimes \cdots \otimes \P_t
  \otimes Q_t
\end{equation}
(as before) and $R \in \mathbb{P}_n$. By linearity, we obtain the
following system after the honest computation, attack, the verifier's
conditional corrections and tracing out of the auxiliary system%
\begin{multline}
\sum_{\substack{a_1,b_1, \ldots a_n, b_n \\ \in \{0,1\}}} \sum_k \sum_{R,R'} \sum_{Q} \alpha_{k,Q \otimes R} \alpha^*_{k,Q\otimes R'}
R  (\xgate^{a_1}\zgate^{b_1} \otimes \cdots \otimes \xgate^{a_n}\zgate^{b_n})\\[-1ex]
   E^{\delta_{P_t}}_{\ell_t}    \tgate_{\ell_t} C_t \ldots  E^{\delta_{P_2}}{\ell_2}    \tgate_{\ell_2} C_2  E^{\delta_{P_1}}_{\ell_1}  \tgate_{\ell_1} C_1 \ket{0^n}
   \bra{0^n} C_1^* \tgate^*_{\ell_1}   {E^*}^{\delta_{P_1}}_{\ell_1} C_2^*   \tgate^*_{\ell_2}   {E^*}^{\delta_{P_2}}_{\ell_2} \ldots\\[1ex]
   \ldots C_t^* \tgate^*_{\ell_t}   {E^*}^{\delta_{P_t}}_{\ell_t}
   (\zgate^{b_1} \xgate^{a_1}\otimes \cdots \zgate^{b_n} \xgate^{a_n})
 R' \otimes {}\\[1ex]
  \egoketbra{a_1, b_1, \ldots ,a_n, b_n}\,.
\end{multline}
We can assume that the verifier then decrypts the system; by the Pauli twirl (\expref{Lemma}{lem:Pauli-twirl}), the terms with $R \neq R'$ vanish, leaving as the computation registers%
\begin{multline}
 \sum_k \sum_{Q,R} \abs{\alpha_{k,Q\otimes R}}^2
R
   E^{\delta_{P_t}}_{\ell_t}    \tgate_{\ell_t} C_t \ldots  E^{\delta_{P_2}}{\ell_2}    \tgate_{\ell_2} C_2  E^{\delta_{P_1}}_{\ell_1}  \tgate_{\ell_1} C_1 \ket{0^n}
 \bra{0^n} C_1^* \tgate^*_{\ell_1}   {E^*}^{\delta_{P_1}}_{\ell_1} C_2^*   \tgate^*_{\ell_2}   {E^*}^{\delta_{P_2}}_{\ell_2} \ldots   C_t^* \tgate^*_{\ell_t}   {E^*}^{\delta_{P_t}}_{\ell_t}
   R\,.
\end{multline}
Denoting $R = R_1 \otimes \cdots \otimes R_n$ ($R_i \in \mathbb{P}_1$), we see that each term with $Q\otimes R$ being benign leads to the output bit having the same distribution as in the honest case
 (Because 
 for all $i$, benign attacks have $\delta_{P_i}=0$, and also $R_n \in \{I, Z\}$ will have no effect on the output qubit that is measured%
 .)
   In the honest case, $p$ is the probability of observing ``0'' (which leads to the verifier accepting). In the case of the Pauli $Q\otimes R$ being \emph{not} benign, we have no bound on the acceptance probability. Thus, with probability at least
\begin{equation}
  (1-p) \sum_k \sum_{Q \in  B_{t,n}} \abs{\alpha_{k, Q}}^2\,,
\end{equation}
the verifier rejects.
\end{proof}

\subsection{Proof of soundness}
\label{sec:soundness-proof}
In order to complete the proof of soundness, we combine \expref{Lemma}{lem:test-run-acceptance} and \expref{Lemma}{lem:acceptance-computation}.
Consider the \expref{Interactive Proof System}{prot:interactive-proof-EPR} for a circuit $C$ on $n$ qubits and with $t$ $\tgate$-gate gadgets, with attack $\{E_k\}$, where
\begin{equation}
  E_k = \sum_Q \alpha_{k,Q} Q\,.
\end{equation}
Let $B_{t,n}$ be the set of benign Pauli attacks, and $B'_{t,n}$ be
the set of non-benign Pauli attacks, and let $p$ be as given in
\expref{Lemma}{lem:acceptance-computation}. Then since a test run
occurs with probability ${2}/{3}$ and a computation run occurs
with probability ${1}/{3}$, the probability that the verifier
rejects is at least%
   \begin{equation}
   \frac{2}{3} \cdot \frac{1}{2}\sum_k\sum_{Q \in B'_{t,n}} \abs{\alpha_{k,Q}}^2 + \frac{1}{3} (1-p) \sum_k\sum_{Q \in B_{t,n}} \abs{\alpha_{k,Q}}^2.
   \end{equation}
 Suppose that the input corresponds to a $\emph{no}$ instance of \textsf{Q-CIRCUIT}. Hence, $1-p \geq {2}/{3}$ and the probability that the verifier rejects is at least ${2}/{9}$ since%
 \begin{align}
 &  \frac{1}{3} \sum_k \sum_{Q \in B'_{t,n}} \abs{\alpha_{k,Q}}^2 + \frac{1-p}{3}  \sum_k\sum_{Q \in B_{t,n}} \abs{\alpha_{k,Q}}^2 \label{Eqn:rej-prob} \\ &\geq
  \frac{1}{3} \sum_k \sum_{Q \in B'_{t,n}} \abs{\alpha_{k,Q}}^2 + \frac{2}{9}  \sum_k\sum_{Q \in B_{t,n}} \abs{\alpha_{k,Q}}^2 \\
&= \frac{2}{9}   \sum_k \sum_{Q\in \mathbb{P}_{2t+n}} \abs{\alpha_{k,Q}}^2 +  \frac{1}{9} \sum_k \sum_{Q \in B'_{t,n}} \abs{\alpha_{k,Q}}^2 \\
& \geq
 \frac{2}{9}\,.
  \end{align}
In \expref{Section}{sec:completeness}, we saw that, in the case that the prover is honest, the probability $c$ of acceptance of a \emph{yes}-instance of \textsf{Q-CIRCUIT} satisfies  $c\geq {8}/{9}$ (completeness). We have determined above that for every prover,  the probability $s$ of acceptance of a \emph{no}-instance satisfies $s\leq 1-{2}/{9} = {7}/{9}$ (soundness). Thus we have $c-s \geq {1}/{9}$, which, by standard amplification,  completes the proof of \expref{Theorem}{thm:main-QPIP}.

One consequence of the proof in this section is to note that the identity attack (\ie, the prover's honest  behaviour)  yields%
\begin{equation}
\sum_k \sum_{Q \in B'_{t,n}} \abs{\alpha_{k,Q}}^2 =0
\end{equation}
and
\begin{equation}
  \frac{1-p}{3}  \sum_k\sum_{Q \in B_{t,n}} \abs{\alpha_{k,Q}}^2 = \frac{1-p}{3}\,,
  \end{equation}
from which we conclude (\expeqref{Equation}{Eqn:rej-prob}) that the identity attack is an attack that minimizes the  probability of rejection of a \emph{no} instance.

\section*{Acknowledgements}
I am grateful to Urmila Mahadev for suggesting the relabelling described in \expref{Section}{section:resources-T-gate}, which simplifies the proof of soundness, and also for related discussions.
It is a pleasure to thank Gus Gutoski for many deep conversations, from which this work originated, and Thomas Vidick for many related discussions.
Furthermore, it is a pleasure to thank  Jacob Krich for supplying background material on quantum simulations. I am also grateful to  Harry Buhrman,  Evelyn Wainewright and to the anonymous referees for feedback on an earlier version of this work.

\bibliographystyle{tocplain}   %
\bibliography{v014a011}

%% Theory of Computing
%% File: v014a011.bbl
%% Publication date: June 11, 2018
%% Authors: Anne Broadbent
% Define empty bibhead if not already defined
\providecommand{\bibhead}[1]{}
% Define tocrefpdfbookmark if not already defined
\expandafter\ifx\csname pdfbookmark\endcsname\relax%
  \providecommand{\tocrefpdfbookmark}{}
\else\providecommand{\tocrefpdfbookmark}{%
   \hypertarget{tocreferences}{}%
   \pdfbookmark[1]{References}{tocreferences}}%
\fi

\tocrefpdfbookmark
\begin{thebibliography}{10}

\bibitem{AA11}\bibhead{AA11}
{\sc Scott Aaronson and Alex Arkhipov}: The computational complexity of linear
  optics.
\newblock {\em Theory of Computing}, 9(4):143--252, 2013.
\newblock Preliminary version in
  \href{http://doi.org/10.1145/1993636.1993682}{STOC'11}.
\newblock [\epfmtdoi{10.4086/toc.2013.v009a004}, \epfmt{arxiv}{1011.3245}]

\bibitem{AA14}\bibhead{AA14}
{\sc Scott Aaronson and Alex Arkhipov}: Boson{S}ampling is far from uniform.
\newblock {\em Quantum Inf. Comput.}, 14(15-16):1383--1423, 2014.
\newblock [\epfmt{arxiv}{1309.7460}, \epfmt{eccc}{TR13-135}]

\bibitem{ABE10}\bibhead{ABE10}
{\sc Dorit Aharonov, Michael {Ben-Or}, and Elad Eban}: Interactive proofs for
  quantum computations.
\newblock In {\em Proc. 1st Symp. Innovations in Computer Science (ICS'10)},
  pp. 453--469, 2010.
\newblock [\epfmt{arxiv}{1704.04487}]

\bibitem{ABEM17}\bibhead{ABEM17}
{\sc Dorit Aharonov, Michael {Ben-Or}, Elad Eban, and Urmila Mahadev}:
  Interactive proofs for quantum computations, 2017.
\newblock [\epfmt{arxiv}{1704.04487}]

\bibitem{AJL06}\bibhead{AJL06}
{\sc Dorit Aharonov, Vaughan Jones, and Zeph Landau}: A polynomial quantum
  algorithm for approximating the {J}ones polynomial.
\newblock {\em Algorithmica}, 55(3):395--421, 2009.
\newblock Preliminary version in
  \href{http://doi.org/10.1145/1132516.1132579}{STOC'06}.
\newblock [\epfmtdoi{10.1007/s00453-008-9168-0},
  \epfmt{arxiv}{quant-ph/0511096}]

\bibitem{AV14}\bibhead{AV14}
{\sc Dorit Aharonov and Umesh Vazirani}: Is {Q}uantum {M}echanics falsifiable?
  {A} computational perspective on the foundations of {Q}uantum {M}echanics.
\newblock In {\em Computability: Turing, G{\"o}del, Church, and Beyond}, pp.
  329--349. MIT Press, 2013.
\newblock [\epfmt{arxiv}{1206.3686}]

\bibitem{AMTW00}\bibhead{AMTW00}
{\sc Andris Ambainis, Michele Mosca, Alain Tapp, and Ronald de~Wolf}: Private
  quantum channels.
\newblock In {\em Proc. 41st FOCS}, pp. 547--553. IEEE Comp. Soc. Press, 2000.
\newblock [\epfmtdoi{10.1109/SFCS.2000.892142},
  \epfmt{arxiv}{quant-ph/0003101}]

\bibitem{Bab85}\bibhead{Bab85}
{\sc L{\'a}szl{\'o} Babai}: Trading group theory for randomness.
\newblock In {\em Proc. 17th STOC}, pp. 421--429. ACM Press, 1985.
\newblock [\epfmtdoi{10.1145/22145.22192}]

\bibitem{BCG+02}\bibhead{BCG+02}
{\sc Howard Barnum, Claude Cr{\'e}peau, Daniel Gottesman, Adam Smith, and Alain
  Tapp}: Authentication of quantum messages.
\newblock In {\em Proc. 43rd FOCS}, pp. 449--458. IEEE Comp. Soc. Press, 2002.
\newblock [\epfmtdoi{10.1109/SFCS.2002.1181969},
  \epfmt{arxiv}{quant-ph/0205128}]

\bibitem{BFKW13}\bibhead{BFKW13}
{\sc Stefanie Barz, Joseph~F. Fitzsimons, Elham Kashefi, and Philip Walther}:
  Experimental verification of quantum computation.
\newblock {\em Nature Physics}, 9(11):727--731, 2013.
\newblock [\epfmtdoi{10.1038/nphys2763}, \epfmt{arxiv}{1309.0005}]

\bibitem{BCG+06}\bibhead{BCG+06}
{\sc Michael {Ben-Or}, Claude Cr\'{e}peau, Daniel Gottesman, Avinatan Hassidim,
  and Adam Smith}: Secure multiparty quantum computation with (only) a strict
  honest majority.
\newblock In {\em Proc. 47th FOCS}, pp. 249--260. IEEE Comp. Soc. Press, 2006.
\newblock [\epfmtdoi{10.1109/FOCS.2006.68}, \epfmt{arxiv}{0801.1544}]

\bibitem{BMP+00}\bibhead{BMP+00}
{\sc P.~Oscar Boykin, Tal Mor, Matthew Pulver, Vwani Roychowdhury, and Farrokh
  Vatan}: A new universal and fault-tolerant quantum basis.
\newblock {\em Inform. Process. Lett.}, 75(3):101--107, 2000.
\newblock [\epfmtdoi{10.1016/S0020-0190(00)00084-3},
  \epfmt{arxiv}{quant-ph/9906054}]

\bibitem{BCC88}\bibhead{BCC88}
{\sc Gilles Brassard, David Chaum, and Claude Cr{\'e}peau}: Minimum disclosure
  proofs of knowledge.
\newblock {\em J. Comput. System Sci.}, 37(2):156--189, 1988.
\newblock [\epfmtdoi{10.1016/0022-0000(88)90005-0}]

\bibitem{Bro15}\bibhead{Bro15}
{\sc Anne Broadbent}: Delegating private quantum computations.
\newblock {\em Canad. J. Physics}, 93(9):941--946, 2015.
\newblock [\epfmtdoi{10.1139/cjp-2015-0030}, \epfmt{arxiv}{1506.01328}]

\bibitem{BFK09}\bibhead{BFK09}
{\sc Anne Broadbent, Joseph~F. Fitzsimons, and Elham Kashefi}: Universal blind
  quantum computation.
\newblock In {\em Proc. 50th FOCS}, pp. 517--526. IEEE Comp. Soc. Press, 2009.
\newblock [\epfmtdoi{10.1109/FOCS.2009.36}, \epfmt{arxiv}{0807.4154}]

\bibitem{BGS13}\bibhead{BGS13}
{\sc Anne Broadbent, Gus Gutoski, and Douglas Stebila}: Quantum one-time
  programs.
\newblock In {\em Proc. 33rd Ann. Intern. Cryptology Conf. (CRYPTO'13)}, pp.
  344--360, 2013.
\newblock [\epfmtdoi{10.1007/978-3-642-40084-1\_20}, \epfmt{arxiv}{1211.1080}]

\bibitem{BJ15}\bibhead{BJ15}
{\sc Anne Broadbent and Stacey Jeffery}: Quantum homomorphic encryption for
  circuits of low {T}-gate complexity.
\newblock In {\em Proc. 35th Ann. Intern. Cryptology Conf. (CRYPTO'15)}, pp.
  609--629, 2015.
\newblock [\epfmtdoi{10.1007/978-3-662-48000-7\_30}, \epfmt{arxiv}{1412.8766}]

\bibitem{CLN05}\bibhead{CLN05}
{\sc Andrew~M. Childs, Debbie~W. Leung, and Michael~A. Nielsen}: Unified
  derivations of measurement-based schemes for quantum computation.
\newblock {\em Phys. Rev. A}, 71(3):032318, 2005.
\newblock [\epfmtdoi{10.1103/PhysRevA.71.032318},
  \epfmt{arxiv}{quant-ph/0404132}]

\bibitem{CGJV17}\bibhead{CGJV17}
{\sc Andrea Coladangelo, Alex Grilo, Stacey Jeffery, and Thomas Vidick}:
  Verifier-on-a-leash: new schemes for verifiable delegated quantum
  computation, with quasilinear resources, 2017.
\newblock [\epfmt{arxiv}{1708.07359}]

\bibitem{CL95}\bibhead{CL95}
{\sc Anne Condon and Richard~E. Ladner}: Interactive proof systems with
  polynomially bounded strategies.
\newblock {\em J. Comput. System Sci.}, 50(3):506--518, 1995.
\newblock Preliminary version in \href{10.1109/SCT.1992.215403}{SCT'92}.
\newblock [\epfmtdoi{10.1109/SCT.1992.215403}]

\bibitem{DFSS05}\bibhead{DFSS05}
{\sc Ivan~B. Damg{\aa}rd, Serge Fehr, Louis Salvail, and Christian Schaffner}:
  Cryptography in the bounded-quantum-storage model.
\newblock {\em SIAM J. Comput.}, 37(6):1865--1890, 2008.
\newblock Preliminary version in
  \href{http://doi.org/10.1109/SFCS.2005.30}{FOCS'05}.
\newblock [\epfmtdoi{10.1137/060651343}, \epfmt{arxiv}{quant-ph/0508222}]

\bibitem{DCEL09}\bibhead{DCEL09}
{\sc Christoph Dankert, Richard Cleve, Joseph Emerson, and Etera Livine}: Exact
  and approximate unitary 2-designs and their application to fidelity
  estimation.
\newblock {\em Phys. Rev. A}, 80(1):012304, 2009.
\newblock [\epfmtdoi{10.1103/PhysRevA.80.012304},
  \epfmt{arxiv}{quant-ph/0606161}]

\bibitem{DFPR14}\bibhead{DFPR14}
{\sc Vedran Dunjko, Joseph~F. Fitzsimons, Christopher Portmann, and Renato
  Renner}: Composable security of delegated quantum computation.
\newblock In {\em 20th Internat. Conf. on the Theory and Appl. of Cryptology
  and Information Security (ASIACRYPT'14)}, pp. 406--425, 2014.
\newblock [\epfmtdoi{10.1007/978-3-662-45608-8\_22}, \epfmt{arxiv}{1301.3662}]

\bibitem{Fey82}\bibhead{Fey82}
{\sc Richard~P. Feynman}: Simulating physics with computers.
\newblock {\em Internat. J. Theoretical Physics}, 21(6-7):467--488, 1982.
\newblock [\epfmtdoi{10.1007/BF02650179}]

\bibitem{FBS+14}\bibhead{FBS+14}
{\sc Kent~A.{\,}G. Fisher, Anne Broadbent, Lynden~K. Shalm, Zhennya Yan,
  Jonathan Lavoie, Robert Prevedel, Thomas Jennewein, and Kevin~J. Resch}:
  Quantum computing on encrypted data.
\newblock {\em Nature Communications}, 5:3074, 2014.
\newblock [\epfmtdoi{10.1038/ncomms4074}, \epfmt{arxiv}{1309.2586}]

\bibitem{FH15}\bibhead{FH15}
{\sc Joseph~F. Fitzsimons and Michal Hajdu{\v{s}}ek}: Post hoc verification of
  quantum computation, 2015.
\newblock [\epfmt{arxiv}{1512.04375}]

\bibitem{FK12}\bibhead{FK12}
{\sc Joseph~F. Fitzsimons and Elham Kashefi}: Unconditionally verifiable blind
  computation, 2012.
\newblock [\epfmt{arxiv}{1203.5217}]

\bibitem{FK17}\bibhead{FK17}
{\sc Joseph~F. Fitzsimons and Elham Kashefi}: Unconditionally verifiable blind
  quantum computation.
\newblock {\em Phys. Rev. A}, 96(1):012303, 2017.
\newblock [\epfmtdoi{10.1103/PhysRevA.96.012303}]

\bibitem{GH05}\bibhead{GH05}
{\sc Zhengting Gan and Robert~J. Harrison}: Calibrating quantum chemistry: A
  multi-teraflop, parallel-vector, full-configuration interaction program for
  the {Cray-X1}.
\newblock In {\em 18th Ann. Conf. on Supercomputing (SC'05)}, pp. 22--22, 2005.
\newblock [\epfmtdoi{10.1109/SC.2005.17}]

\bibitem{GKAE13}\bibhead{GKAE13}
{\sc Christian Gogolin, Martin Kliesch, Leandro Aolita, and Jens Eisert}:
  Boson-{S}ampling in the light of sample complexity, 2013.
\newblock [\epfmt{arxiv}{1306.3995}]

\bibitem{GKR08}\bibhead{GKR08}
{\sc Shafi Goldwasser, Yael~Tauman Kalai, and Guy~N. Rothblum}: Delegating
  computation: {I}nteractive proofs for muggles.
\newblock {\em J. ACM}, 62(4):27:1--27:64, 2015.
\newblock Preliminary version in
  \href{http://doi.org/10.1145/1374376.1374396}{STOC'08}.
\newblock [\epfmtdoi{10.1145/2699436}, \epfmt{eccc}{TR17-108}]

\bibitem{GMR89}\bibhead{GMR89}
{\sc Shafi Goldwasser, Silvio Micali, and Charles Rackoff}: The knowledge
  complexity of interactive proof systems.
\newblock {\em SIAM J. Comput.}, 18(1):186--208, 1989.
\newblock Preliminary version in
  \href{http://doi.org/10.1145/22145.22178}{STOC'85}.
\newblock [\epfmtdoi{10.1137/0218012}]

\bibitem{GC99}\bibhead{GC99}
{\sc Daniel Gottesman and Isaac~L. Chuang}: Demonstrating the viability of
  universal quantum computation using teleportation and single-qubit
  operations.
\newblock {\em Nature}, 402(6760):390--393, 1999.
\newblock [\epfmtdoi{10.1038/46503}, \epfmt{arxiv}{quant-ph/9908010}]

\bibitem{HH16}\bibhead{HH16}
{\sc Masahito Hayashi and Michal Hajdu{\v{s}}ek}: Self-guaranteed
  measurement-based quantum computation, 2016.
\newblock [\epfmt{arxiv}{1603.02195}]

\bibitem{HM15}\bibhead{HM15}
{\sc Masahito Hayashi and Tomoyuki Morimae}: Verifiable measurement-only blind
  quantum computing with stabilizer testing.
\newblock {\em Phys. Rev. Lett.}, 115(22):220502, 2015.
\newblock [\epfmtdoi{10.1103/PhysRevLett.115.220502},
  \epfmt{arxiv}{1505.07535}]

\bibitem{JJUW10}\bibhead{JJUW10}
{\sc Rahul Jain, Zhengfeng Ji, Sarvagya Upadhyay, and John Watrous}: {\sf QIP =
  PSPACE}.
\newblock {\em Comm. ACM}, 53(12):102--109, 2010.
\newblock Preliminary version in
  \href{http://doi.org/10.1145/1806689.1806768}{STOC'10}.
\newblock [\epfmtdoi{10.1145/1859204.1859231}, \epfmt{arxiv}{0907.4737}]

\bibitem{KDK15}\bibhead{KDK15}
{\sc Theodoros Kapourniotis, Vedran Dunjko, and Elham Kashefi}: On optimising
  quantum communications in verifiable quantum computing.
\newblock In {\em Asian Quantum Info. Sci. Conf. (AQIS'15)}, pp. 23--25, 2015.
\newblock [\epfmt{arxiv}{1506.06943}]

\bibitem{KW15}\bibhead{KW15}
{\sc Elham Kashefi and Petros Wallden}: Optimised resource construction for
  verifiable quantum computation.
\newblock {\em J. Physics A}, 50(14):145306, 2017.
\newblock [\epfmtdoi{10.1088/1751-8121/aa5dac}, \epfmt{arxiv}{1510.07408}]

\bibitem{KKMV08}\bibhead{KKMV08}
{\sc Julia Kempe, Hirotada Kobayashi, Keiji Matsumoto, and Thomas Vidick}:
  Using entanglement in quantum multi-prover interactive proofs.
\newblock {\em Comput. Complexity}, 18(2):273--307, 2009.
\newblock Preliminary version in
  \href{http://doi.org/10.1109/CCC.2008.6}{CCC'08}.
\newblock [\epfmtdoi{10.1007/s00037-009-0275-3}, \epfmt{arxiv}{0711.3715}]

\bibitem{LFKN90}\bibhead{LFKN90}
{\sc Carsten Lund, Lance Fortnow, Howard Karloff, and Noam Nisan}: Algebraic
  methods for interactive proof systems.
\newblock {\em J. ACM}, 39(4):859--868, 1992.
\newblock Preliminary version in
  \href{http://doi.org/10.1109/FSCS.1990.89518}{FOCS'90}.
\newblock [\epfmtdoi{10.1145/146585.146605}]

\bibitem{NC00}\bibhead{NC00}
{\sc Michael~A. Nielsen and Issac~L. Chuang}: {\em Quantum Computation and
  Quantum Information}.
\newblock Cambridge University Press, 2000.
\newblock [\epfmtdoi{10.1017/CBO9780511976667}]

\bibitem{RUV13}\bibhead{RUV13}
{\sc Ben~W. Reichardt, Falk Unger, and Umesh Vazirani}: Classical command of
  quantum systems.
\newblock {\em Nature}, 496(7446):456--460, 2013.
\newblock Preliminary version in
  \href{http://doi.org/10.1145/2422436.2422473}{ITCS'13}.
\newblock [\epfmtdoi{10.1038/nature12035}, \epfmt{arxiv}{1209.0448},
  \epfmt{arxiv}{1209.0449}]

\bibitem{Sha92}\bibhead{Sha92}
{\sc Adi Shamir}: {\sf IP = PSPACE}.
\newblock {\em J. ACM}, 39(4):869--877, 1992.
\newblock Preliminary version in
  \href{http://doi.org/10.1109/FSCS.1990.89519}{FOCS'90}.
\newblock [\epfmtdoi{10.1145/146585.146609}]

\bibitem{SP00}\bibhead{SP00}
{\sc Peter~W. Shor and John Preskill}: Simple proof of security of the {BB}84
  quantum key distribution protocol.
\newblock {\em Phys. Rev. Lett.}, 85(2):441--444, 2000.
\newblock [\epfmtdoi{10.1103/physrevlett.85.441},
  \epfmt{arxiv}{quant-ph/0003004}]

\bibitem{SVB+14}\bibhead{SVB+14}
{\sc Nicol{\`o} Spagnolo, Chiara Vitelli, Marco Bentivegna, Daniel~J Brod,
  Andrea Crespi, Fulvio Flamini, Sandro Giacomini, Giorgio Milani, Roberta
  Ramponi, Paolo Mataloni, Roberto Osellame, Ernesto~F. Galv{\~a}o, and Fabio
  Sciarrino}: Experimental validation of photonic boson sampling.
\newblock {\em Nature Photonics}, 8(8):615--620, 2014.
\newblock [\epfmtdoi{10.1038/nphoton.2014.135}, \epfmt{arxiv}{1311.1622}]

\bibitem{Wat03}\bibhead{Wat03}
{\sc John Watrous}: {\sf PSPACE} has constant-round quantum interactive proof
  systems.
\newblock {\em Theoret. Comput. Sci.}, 292(3):575--588, 2003.
\newblock Preliminary version in
  \href{http://doi.org/10.1109/SFFCS.1999.814583}{FOCS'99}.
\newblock [\epfmtdoi{10.1016/S0304-3975(01)00375-9}, \epfmt{arxiv}{cs/9901015}]

\bibitem{ZLC00}\bibhead{ZLC00}
{\sc Xinlan Zhou, Debbie~W. Leung, and Isaac~L. Chuang}: Methodology for
  quantum logic gate construction.
\newblock {\em Phys. Rev. A}, 62(5):052316, 2000.
\newblock [\epfmtdoi{10.1103/PhysRevA.62.052316},
  \epfmt{arxiv}{quant-ph/0002039}]

\end{thebibliography}

\begin{tocauthors}
\begin{tocinfo}[broadbent]
 Anne Broadbent\\
Associate professor\\
 University of Ottawa\\
 Ottawa, ON, Canada\\
 abroadbe\tocat{}uottawa\tocdot{}ca \\
 \url{http://mysite.science.uottawa.ca/abroadbe/}
\end{tocinfo}
\end{tocauthors}

\begin{tocaboutauthors}
\begin{tocabout}[broadbent]
\textsc{Anne Broadbent} holds an undergraduate degree  (2002) in
\href{https://uwaterloo.ca/combinatorics-and-optimization/}{Combinatorics and Optimization} from the
\href{https://uwaterloo.ca/}{University of Waterloo}.
Her Master's (2004) and \phd\ (2008) work were supervised by Gilles Brassard and Alain Tapp at the
\href{http://diro.umontreal.ca/}{D\'epartement d'informatique et de recherche op\'erationnelle} of
\href{http://www.umontreal.ca/}{Universit\'e de Montr\'eal}.
In her spare time, she likes researching
\href{https://doi.org/10.1007/s10623-015-0157-4}{quantum cryptography beyond quantum key distribution}.
\end{tocabout}
\end{tocaboutauthors}

\end{document}